\documentclass[12pt]{article}
\usepackage{authblk}

\usepackage{amsmath}
\usepackage{amsfonts}
\usepackage{amsthm}
 \usepackage{hyperref}
\usepackage{bbm}
\usepackage{empheq}
\usepackage{scrextend}
\usepackage{mathtools}
\usepackage{amssymb}
\usepackage{graphicx}
\graphicspath{ {./images/} }
\newcommand*{\p}{\mathbb{P}}
\usepackage{xcolor}
\usepackage{dsfont}
\usepackage{natbib}
\usepackage{graphicx}
\usepackage{cases}
\graphicspath{ {./images/} }
\usepackage{float}
\usepackage{multibib}
\usepackage{algorithm} 
\usepackage{setspace}
\usepackage{booktabs}
\usepackage[utf8]{inputenc}

\usepackage{accents}

\newcommand{\norm}[1]{\left\lVert#1\right\rVert}

\newtheorem{Lemma}{Lemma}[section]

\newtheorem{theorem}[Lemma]{Theorem}

\newtheorem{remark}[Lemma]{Remark}
\newtheorem{assumption}[Lemma]{Assumption}

\theoremstyle{definition}

\date{}

\definecolor{darkblue}{rgb}{.1, 0.1,.8}
\definecolor{darkgreen}{rgb}{0,0.8,0.2}
\definecolor{darkred}{rgb}{.8, .1,.1}

\mathtoolsset{showonlyrefs}
\newcommand*{\E}{\mathbb{E}}

\newcommand*{\N}{\mathbb{N}}
\newcommand*{\Z}{\mathbb{Z}}
\newcommand*{\R}{\mathbb{R}}
\renewcommand{\P }{{\mathbb P}}
\newcommand{\1}{\mathbbm{1}}

\title{Multiscale detection of practically significant changes in a gradually varying time series}

\begin{document}

\author{
  { Patrick Bastian, ~Holger Dette} \\
{ Ruhr-Universit\"at Bochum} \\
{ Fakult\"at f\"ur Mathematik} \\
{ 44780 Bochum, Germany} 
 }

\maketitle
\begin{abstract}
In many change point problems it is reasonable to assume that compared to a benchmark at a given time point $t_0$ the properties of the observed stochastic process   change gradually over time for $t >t_0$.
Often, these gradual changes are not of interest as long as they are small (nonrelevant), but one  
is interested  in the question if 
the deviations are practically significant in the sense that the deviation of the process compared to  the time $t_0$ (measured by an appropriate metric) exceeds a given threshold, which is of practical significance (relevant change). 

In this paper we develop novel and powerful change point analysis  for detecting such deviations  in a sequence of gradually varying means, which is compared with the  average mean  from a previous time period.    Current approaches to this problem suffer from low power, rely on the  selection of  smoothing parameters  and  require a rather regular (smooth) development for the means.  We develop a multiscale procedure that alleviates all these issues, validate it theoretically and demonstrate its good finite sample performance on both synthetic and real data.  
\end{abstract}

\section{Introduction}

   \def\theequation{1.\arabic{equation}}	
   \setcounter{equation}{0}

Change point analysis is an ubiquitous topic in mathematical statistics with numerous applications in diverse areas such as economics, climatology and linguistics.    In this paper we are concerned with the detection of (gradual) changes in  the  univariate time series 
\begin{align}\label{hd1}
    X_i=\mu(i/n)+\epsilon_i~, ~~~i=1, \ldots , n , 
\end{align}
where $(\epsilon_i)_{1 \leq i \leq n}$ is a centered error process and $\mu:[0,1]\rightarrow \R$ is an unknown  (mean) function. 
Starting with the seminal work of  
\citep*{Page1955}, a  large part of the literature considers the case where the function 
$\mu$  is piecewise constant with at most one change and we refer to  reviews of \citep*{Aue2013},   \citep*{Aue2023} and to the recent textbook of \citep*{Horvath24} and the reference therein.
More recently  the problem of detecting  multiple changes has found  considerable interest as well. Among many others, we  mention \citep*{Fryzlewicz2013}, \citep*{Frick2014},  \citep*{baranowskietal2019}, \citep*{Eckle2020}, who considered  one-dimensional data and to \citep*{CWS22,LWR23,Padilla_Yu_Wang_Rinaldo_multivariate_nonparametric}
for some recent results in the high-dimensional and multivariate case. A good  review  of the current state of the art on data segmentation of a piecewise constant signal can be found in the recent papers of  \citep*{TRUONG2020} and \citep*{cho:kirch:2024}. A common  aspect of the methodology  in  most of  these references consists in the fact that it is based on 
the construction of testing procedures for the hypothesis
\begin{align} \label{hd2}
    H_0:\mu(t)=c \quad \forall t\in [0,1],
\end{align}
where $c$ is an unknown constant, and the different procedures address different forms of the piecewise constant function $\mu$ under the alternative. Here a large part of the literature has its focus on piecewise stationary alternatives.

While there are many applications where this  can be well  justified,  at least approximately 
 \citep[see  for instance][for some examples]{Aston2012,Hotz2013},
 there exist also many situations  where it  is not  reasonable to assume that the mean function $\mu$ in model \eqref{hd1} is  piecewise constant 
 over the full period because it is continuously   smoothly between potential jumps points. 
   Examples include climate data \citep*{Karl1995}, where it is continuously varying with no jumps, financial data \citep*{Vogt2015} and medical data \citep*{Gao2018}. 
  In this context,
  \citep*{Muller1992}, \citep*{Gijbels2007} and \citep*{Gao2008}  among others, considered model \eqref{hd1} with a gradually  varying mean function and  sudden jumps
 and developed change point analysis for determining 
the locations of these jumps. 
 \citep*{Vogt2015} considered the problem of detecting gradual changes in a more general context. For the special case of the  mean they
assumed that $\mu $  is  constant for some time, then
slowly starts to change (with no jumps), and  
developed a fully
nonparametric method to estimate such a smooth change point.

The present paper differs from these works.  Although we consider a model of the form \eqref{hd1} with  
smooth regression function and potential jumps, we are not interested in the locations of the jump points 
or in the first time point where the mean functions starts to vary. Instead  we define a  change  differently as a practical significant deviation of $\mu$ in the interval $[0,1]$ from  a given  benchmark, say $\mu_0^{t_0}$. More precisely, for a constant $\mu_0^{t_0}$  and a threshold $\Delta >0 $ we are interested in the hypotheses
\begin{align}
\label{p9}
    H_0(\Delta): \sup_{t \in [t_0,1]}|\mu(t)-\mu_0^{t_0}|\leq \Delta \ \ vs \ \ H_1(\Delta):\sup_{t \in [t_0,1]}|\mu(t)-\mu_0^{t_0}|> \Delta , 
\end{align}
where $t_0\in [0,1)$ is a given point defining the time interval of interest.
A prominent example for the consideration of these hypotheses  is the analysis of global mean temperature anomalies, where one is interested in a significant 
deviation of the current temperatures from a reference value $\mu_0^{t_0}$ at time $t_0$ (such as the average temperature before time $t_0$), and  an important problem is to  investigate if these deviations exceed 
$\Delta$ after the time $t_0$ , such as  $1.5 $ degrees
Celsius as postulated in the Paris agreement. 
Hypotheses of the form \eqref{p9} are also considered in  quality control (often in an online framework), where one is interested in the ``stability'' of  a given process. This means that the sequence of means stays within a predefined range
(as specified by the null hypothesis in \eqref{p9}). 
While in change-point analysis the focus is often on testing
for the presence of a change and on estimating the time at which a change occurs
once it has been detected, quality control  typically has its focus more on
detecting such a change as quickly as possible after it occurs \citep[see, for example,][]{woodmont1999}.
  
Despite their importance, a  test for the hypotheses in  \eqref{p9} has only been recently developed by  \citep*{buecher21}, who used 
local linear regression techniques  to  estimate  the quantity  $\sup_{t \in [t_0,1]}|\mu(t)-\mu_0^{t_0} |$. They proposed to reject the null hypothesis for large values of the estimate, where critical values are either obtained by asymptotic theory (which shows that a  properly scaled version of the estimate converges weakly to a  Gumbel type distribution) or by resampling based on a Gaussian approximation.
As a consequence, the resulting procedure suffers from several deficits
First, it is conservative in finite samples, particularly for small sample sizes. Second, the mean function $\mu$ in model \eqref{hd1} has to be twice differentiable and the difference $\mu(t)-\mu_0^{t_0}$ has to satisfy some convexity properties, making the method unreliable for less smooth or discontinuous functions.  Finally, the procedure proposed in \citep*{buecher21}  relies on a bandwidth parameter that also prevents detection of local alternatives at the standard parametric rate $n^{-1/2}$.

\textbf{Our contribution} in this paper  is a novel multiscale test for the hypotheses \eqref{p9} that does not suffer from the aforementioned problems. More precisely, we compare (an estimate of)  the  benchmark with local means  calculated over different scales and 
establish the weak convergence of this statistic (see Theorem \ref{Thm:Nulldist}). The limit is not  
distribution free and  depends on sequences of ``extremal sets'',
which can only be defined  implicitly due to certain properties of the Brownian motion. 
Based on this result, we  propose  first a testing procedure that relies on a distributional upper bound of the limiting distribution 
and only requires estimation of a long-run variance. By construction, the resulting test is conservative, but it
it  can detect local alternatives at a parametric rate and already 
 outperforms the existing methodology 
 in our finite sample study. By estimating the extremal sets we are able to  construct a more elaborate procedure which achieves the nominal level asymptotically.  The resulting test can 
detect local alternatives at a parametric rate as well and yields a further substantial improvement  of the finite sample performance. 

Summarizing,  both novel  multiscale tests for practically relevant changes in the gradually changing mean function can detect 
local alternatives converging to the null at a faster rate than the currently available procedure. In contrast to this test, they 
are  applicable for piecewise smooth mean functions without further constraints on their geometry and do not require the choice of smoothing parameters.

\section{Multiscale detection of relevant changes}
\label{sec}

   \def\theequation{2.\arabic{equation}}	
   \setcounter{equation}{0}

In the following we consider the location scale model
\begin{align}
\label{model}
    X_{i}=\mu(i/n)+\epsilon_i \quad i=1,...,n
\end{align}
where $(\epsilon_i)_{i \in \N}$ is a stationary centered process and $\mu$ is a bounded and piecewise Lipschitz-continuous function. We are interested in detecting significant deviations of the function $\mu$ on the interval $[t_0,1]$
from its long-term average in the past
\begin{align} \label{hd6}
\mu_0^{t_0}:={1 \over t_0} \int_0^{t_0}\mu(x)dx~,
\end{align}
where $0<t_0<1$ is some predefined point in time. We address this problem by testing the hypotheses 
\begin{align}
\label{hypotheses}
H_0(\Delta)\colon d_\infty \leq \Delta \quad vs \quad H_1(\Delta)\colon d_\infty>\Delta~, 
\end{align}
where 
\begin{align} \label{hd3}
    d_\infty=\sup_{t \in [t_0,1]}|\mu(t)-\mu_0^{t_0}|
\end{align}
denotes the maximum (absolute) deviation of the function $\mu$ from  $\mu_0^{t_0}$ over the interval $[t_0,1]$ and $\Delta>0$ is a given threshold.

 A typical application for this benchmark  is encountered in the analysis of temperature data where deviations from a pre-industrial average temperature are of interest. In this case the 
 choice of threshold is also quite clear. For example, if we are interested in the 
global mean temperature anomalies,
a reasonable choice is 
 $\Delta=1.5$ degrees Celsius corresponding to the Paris Agreement adopted at the UN Climate Change Conference(COP21) in Paris, 2015. In other circumstances the choice of 
  $\Delta$ might not be so obvious and has to be carefully discussed for each application.  We defer further discussion of this issue to Remark \ref{DeltaChoice}, where we also propose a data based choice of the threshold $\Delta$.

In the remainder of this section we introduce the necessary concepts and assumptions to define and theoretically analyze a multiscale test statistic for the testing  problem \eqref{hypotheses}.

\begin{assumption}
The random variables
in model \eqref{model}
form a triangular array of real valued random variables where $(\epsilon_i)_{i \in \Z}$ is a mean zero stationary sequence with existing  long run variance $\sigma^2  = \sum_{i \in \Z}\E[\epsilon_0\epsilon_i]$
    \begin{enumerate}
        \item[ (A1)] For  some $0<p<1/2$ there exists a standard Brownian motion $B$, such that for all $k \in \N$ \begin{align}
            \Big| \sum_{i=1}^k\epsilon_i - \sigma B(k) \Big|\leq Ck^{1/2-p} 
        \end{align}    
        almost surely for some constant $C>0$.        
        \item[(A2)] The function $\mu:[0,1]\rightarrow \R$ is piecewise Lipschitz continuous with finitely many jumps.        
    \end{enumerate}
\end{assumption}

Assumption  (A1) is a high level assumption that is standard in the literature and is satisfied by a large class of weakly dependent time series (see, for instance \citep*{Dehling1983} for mixing, \citep*{Wu2005} for physically dependent and \citep*{Berkes2011} for $L^p$-$m$-approximable processes). Assumption (A2) is a weak regularity assumption on the function $\mu$, that can in principle be weakened to Hölder continuity with some additional technical effort. 

We now introduce the test statistic, and to that end denote for $j<k$ by
\begin{align}
    \hat \mu_j^k=\frac{1}{k-j}\sum_{i=j+1}^kX_i
\end{align}
the (local)  mean of the observations $X_{j+1}, \ldots , X_k$. Note that 
\begin{align}
 \mathbb{E} \big [   \hat \mu_j^k  \big ]  =   \frac{1}{k-j}\sum_{i=j+1}^k\mu(i/n)\simeq \frac{n}{k-j}\int_{j/n}^{k/n}\mu(t)dt~ 
\end{align} 
and therefore  $\hat \mu_0^{\lfloor nt_0 \rfloor } - \hat \mu_j^k  $
(approximately) compares the integral of the function $\mu$ over the interval $[0, \lfloor nt_0 \rfloor /n ]$
with the ``local'' integral over the interval $[j/n, k/n]$, which is approximately given by  $\frac{n}{k-j}\int_{j/n}^{k/n}\mu(t)dt \approx \mu(\frac{k+j}{2n})$ if $k-j$ is small. We now consider these differences on different scales and 
 define for a sequence $(c_n)_{n \in \N}$ of natural numbers  such that $c_n \rightarrow \infty$ and
\begin{align}
\label{p1}
    \frac{n^{1-2p}}{c_n}=o(1),
\end{align}
the test statistic
\begin{align}
\label{teststat}
    \hat T_{n,\Delta}=\sup_{\substack{c_n \leq c \\ c \leq n-\lfloor nt_0 \rfloor \\ c \in \N }}\sup_{\substack{|k-j|=c \\ k>j\geq \lfloor nt_0 \rfloor}}\Big( \sqrt{c} \big|\hat \mu_0^{\lfloor nt_0 \rfloor}-  \hat \mu_j^k  \big|-\sqrt{2\log\Big(\frac{ne}{c}\Big)}-\sqrt{c}\Delta\Big).
\end{align}
 By the discussion of the previous paragraph  $ \hat \mu_j^k$ estimates 
$\mu(\frac{k+j}{2n})$ for  smaller scales, which ensures that relatively short excursions of the function $ t \to  |\mu (t) -\mu_0^{t_0}|$ above $\Delta$ are detected.   For larger scales the statistic \eqref{teststat} is able to take advantage of longer excursions of the function $ t \to  |\mu(t)-\mu_0^{t_0}|$ above the threshold $\Delta$, thereby increasing the power of the test substantially.  The additive factor
\begin{align}   \Gamma_n(c):=\sqrt{2\log\Big(\frac{ne}{c}\Big)}
\end{align} 
equalizes the magnitude of the different scales which would otherwise be dominated by the small scales. We also note that the scaling factor $\sqrt{c}$ depends on $n$  and is of larger and smaller  order than $n^{1/2-p}  $ and $n^{1/2}$, which leads to a non-trivial  asymptotic distribution of the statistic $\hat T_{n,\Delta}$ in the case $d_\infty=\Delta$, which we call \textit{boundary} of the hypotheses. Our first main result provides such a weak convergence result for the statistic

\begin{align}
    \label{hd4}
          \hat T_n =\sup_{\substack{c_n \leq c \leq n-\lfloor nt_0 \rfloor \\ c \in \N }}\sup_{\substack{|k-j|=c \\ k>j\geq \lfloor nt_0 \rfloor}}\Big( \sqrt{c} \big|\hat \mu_0^{\lfloor nt_0 \rfloor}-\hat \mu_j^k  \big|-\Gamma_n(c)-d_{\infty}  \Big) , 
\end{align}
which reduces to the statistic $ \hat T_{n,\Delta} $ in \eqref{teststat}, if the centering  term 
$d_\infty$ defined in \eqref{hd3} is replaced by the threshold $\Delta$.

\begin{theorem}
\label{Thm:Nulldist}
    Grant assumptions  (A1) and (A2).   We then have 
 \begin{align}
     \label{hd11}
       \hat T_n
       \overset{d}{\longrightarrow} T_{d_\infty}
 \end{align}
     where    
     \begin{align}   
      \label{p2}
      T_{d_\infty}  :=\sigma \lim_{\epsilon \downarrow 0}\sup_{(s,t) \in \mathcal{A}_{\epsilon,d_\infty}} \Big \{ \mathfrak{s}(s,t)\Big(\sqrt{t-s}\frac{B(t_0)}{t_0}-\frac{ B(t)- B(s)}{\sqrt{t-s}}\Big)- \Gamma(t-s) \Big \} , 
    \end{align}   
     $B$ denotes a standard Brownian motion 
    and 
    \begin{align}
    \label{det1}
        \mathcal{A}_{\epsilon,d_\infty}&=\Big\{ (s,t) \in [t_0,1]^2 ~\Big | ~s<t, ~\big|\mu_0^{t_0}-\mu_s^t\big|\geq d_\infty-\epsilon\Big\}\\
        \mu_s^t&=\frac{1}{t-s}\int_s^t\mu(x)dx\\ \mathfrak{s}(s,t)&  =\text{sgn}(\big(\mu_0^{t_0}-\mu_s^t\big)\\
        \Gamma(t-s)&=\sqrt{2\log\Big(\frac{e}{t-s}\Big)}~. 
        \label{det3}
    \end{align}
    Moreover, the distribution of the random variable $T_{d_\infty}$ is continuous. In particular,
  \begin{itemize}
      \item[(1)] if  $d_\infty=\Delta$, we have
        $\hat T_{n,\Delta}\overset{d}{\longrightarrow} T_{\Delta},$
      \item[(2)]
         If $d_\infty<\Delta$ or $  d_\infty>\Delta$ we have 
         $\hat T_{n,\Delta}\overset{\p}{\longrightarrow}   -\infty$ or $\hat T_{n,\Delta}\overset{\p}{\longrightarrow}  \infty$, respectively. 
       \end{itemize}     
\end{theorem}

Note that   the limit distribution in Theorem \ref{Thm:Nulldist} is not distribution free even if the long run variance of the error process $(\epsilon_i)_{i \in \N}$ would be known. In fact this distribution   depends in a rather delicate way on the 
 regression function $\mu$ which appears in the definition $T_{d_\infty} $  
through the set ${\cal A}_{\varepsilon,{d_\infty}}$. 
This is contrast to many other  multiscale tests proposed in the literature  
\citep[see, for example,][] {Duembgen2001, Duembgen2008, Schmidthieber2013,Eckle2020}. The difference  can be explained by the fact that  these and - to the best of our knowledge - all other papers on multiscale testing  do not consider relevant hypotheses of the form \eqref{p9} and \eqref{hd3} with $\Delta >0$. In fact, transferring  the hypotheses considered in the multiscale testing literature so far to the situation considered in this paper yields the testing problem for 
the  ``classical''  hypotheses 
\begin{align*}
H_0:d_\infty = 0 \quad vs \quad H_1:d_\infty >0 ~,
\end{align*}
which  corresponds to the choice  $\Delta =0$ in \eqref{hypotheses}. 
It follows from the arguments   given in the proof  of Theorem \ref{Thm:Nulldist}  that in the case $d_\infty =0$
$$
 \hat T_n 
       \overset{d}{\longrightarrow} \sigma M~,
       $$
where the random variable $M$ is defined by
\begin{align}
\label{p4}
    M= \ \sup_{t_0\leq s < t \leq 1}\Big|\sqrt{t-s}\frac{B(t_0)}{t_0}-\frac{B(t)-B(s)}{\sqrt{t-s}}\Big|-\Gamma(t-s).
\end{align}

As the quantiles of the distribution of $M$ can be obtained by simulation, we can already  use this result for the construction of a  valid (conservative) inference procedure for the hypotheses \eqref{hypotheses}
employing the  upper (distributional) bound  
\begin{align}
    \label{dette3} 
  \mathbb{P} (T_\Delta > q)  \leq \mathbb{P} ( \sigma M > q) ~
\end{align}
and estimating the long run variance $\sigma^2$. 
 To be specific,  following \citep*{Wu2007} we  define
\begin{align} \label{det4}
    \hat \sigma^2 = \frac{1}{\lfloor n/m\rfloor -1}\sum_{j=1}^{\lfloor n/m \rfloor -1}\frac{\big( \hat \mu_{(j-1)m}^{jm}-\hat \mu_{jm}^{(j+1)m}\big)^2}{2m}
\end{align}
as an estimator of the long run variance, 
where the parameter $m  \in \N $ converges to $\infty$ as $n\to \infty$  and is  is proportional to $n^{1/3}$.  The null hypothesis is then rejected, whenever
 \begin{align}
    \label{test1}
       T_{n,\Delta} \geq \hat \sigma q_{1-\alpha}~,
    \end{align}
where  $q_{1-\alpha}$ denotes the $(1-\alpha)$ quantile of the distribution of $M$. We will show in the  Section \ref{sec5.4} of the appendix that (under the assumptions made in this paper) 
 the estimator 
$\hat \sigma^2$ is  consistent for  $\sigma^2$, that is, 
\begin{align}
    \hat \sigma^2=\sigma^2+O_\p(n^{-1/3})~,
\end{align}
which yields the following result.
\begin{theorem}
 Under assumptions  (A1) and (A2) the test defined by   \eqref{test1} is  consistent and has  asymptotic level $\alpha$.
\end{theorem}

We continue investigating the  asymptotic power properties of the test \eqref{test1} by considering a class of local alternatives of the form
\begin{align}
\label{p10}
    \mu_n(t)-\mu_0^{t_0}=\Delta+\beta_nh(t)~,
\end{align}
where $h$ is some non-negative and Lipschitz continuous function.
\begin{theorem}
 \label{Thm:LocAlt}
Let assumptions   (A1) and (A2) and be satisfied and consider local alternatives of the form \eqref{p10} with $\beta_n=n^{-1/2}$, then

     \begin{align}
\label{LocAlt}
    \hat T_{n,\Delta}\overset{d}{\rightarrow}\sigma\sup_{t_0\leq s<t\leq 1}\Big( \sqrt{t-s}\frac{B(t_0)}{t_0}-\frac{B(t)-B(s)}{\sqrt{t-s}}-\Gamma(t-s)+    
    {1 \over \sqrt{t-s}}\int_s^th(x)dx \Big)
\end{align}
 \end{theorem}

We collect some observations that follow from this  result in the following remark.
 \begin{remark} ~~~~
    {\rm 
    \begin{itemize}
        \item[(1)] By  Theorem \ref{Thm:LocAlt}  the test 
        \eqref{test1} 
can detect local alternatives converging to the null hypothesis at a parametric rate $n^{-1/2}$. This establishes a substantial improvement over the results obtained in \citep*{buecher21}, where the nonparametric rate
            \begin{align}
                \beta_n\simeq \Big(\sqrt{nh_n\log(h_n)}\Big)^{-1}
            \end{align}
        is required to obtain non-trivial power. Here $h_n$ is the bandwidth used for 
        the local linear estimator of the regression function $\mu$.
   \item[(2)] It 
is clear from the inequality \eqref{dette3}    that the test \eqref{test1} is  in general conservative even in the case that $|\mu(t)-\mu_0^{t_0}|=\Delta$ for all $t \in [t_0,1]$.
Comparing the definitions of the random variables $T_\Delta$ and $M$ in \eqref{p2} and \eqref{p4}, respectively, the difference $\mathbb{P} (T_\Delta > q)  -  \mathbb{P} (\sigma M > q)$ will be large for those models where   $d_\infty=\Delta$ and where at the same time the set  $ \{ s \in [t_0,1] ~|~ |\mu_0^{t_0}-\mu(s)| <\Delta \} $ is ``large''.  This 
indicates  the need for a test procedure that is able to take into account the structure of the 
        sets ${\cal A}_{\varepsilon,d_\infty}$ appearing in the definition of the random variable $T_{d_\infty}$ in \eqref{hd11}.     
    \end{itemize}
}  
\end{remark}
    
     To alleviate the issue raised in the second part of the previous remark we will develop an alternative test 
which uses quantiles from a distribution which approximates the distribution of the random variable $T_\Delta$ more directly. Obviously, such an approach  has to take the estimation of the sets $
        \mathcal{A}_{\epsilon,d_\infty}$ in \eqref{det1} 
        into account.
For this purpose  we define  for a scale parameter $c \in \N$
\begin{align}
    \hat{\mathcal{E}}_c=\Big\{ (j,k) \in \{\lfloor nt_0 \rfloor,\ldots ,n\}^2 ~\Big|~ k-j=c, ~\big|\hat \mu_0^{\lfloor nt_0 \rfloor}-\hat\mu_{j}^{k}\big|\geq \hat d_{\infty,c}-\hat \sigma \log(n)/\sqrt{c} \Big\} ~
\end{align}
as an estimator of the extremal set at scale $c$,  where 
$$
\hat d_{\infty,c}=\max_{k-j=c, \lfloor nt_0 \rfloor\leq j<k \leq n}|\hat \mu_0^{\lfloor nt_0 \rfloor}-\hat \mu_j^k|.
$$
Note that the sequence of sets
\begin{align}
    \hat{\mathcal{A}}_n=\bigcup_{c_n \leq c \leq n-\lfloor nt_0 \rfloor}\hat{\mathcal{E}}_c
\end{align}
may heuristically be interpreted as a sequence of estimators for the sets $\mathcal{A}_{\epsilon,d_\infty}$ for a suitable sequence of $\epsilon\downarrow 0$. Next we introduce 
\begin{align}    
    \hat {\mathfrak{s}}(j/n,k/n)&=\text{sgn}((\hat \mu_0^{\lfloor nt_0 \rfloor}-\hat \mu_j^k)
\end{align}
as an estimator of  the sign $\mathfrak{s}(j/n,k/n)$ in \eqref{det3} and consider the 
statistic  
\begin{align}
    \hat T_n^*= \hat \sigma \sup_{\substack{c_n \leq c \leq n-\lfloor nt_0 \rfloor \\ c \in \N }}\sup_{(j,k) \in \hat {\mathcal{E}}_c}\hat{\mathfrak{s}}(j/n,k/n)\Big(\sqrt{c}\frac{B(\lfloor nt_0 \rfloor)}{\lfloor nt_0 \rfloor}-\frac{B(k)-B(j)}{\sqrt{c}}\Big)- \Gamma_n(c)    
\end{align}
where $\hat \sigma^2$ is the estimator of the long run variance defined in \eqref{det4} and suprema over empty sets are defined as $0$. We denote 
the $(1-\alpha)$-quantile of the distribution of $\hat T_n^*$ by  $q_{1-\alpha}^*$, which  can  easily be  simulated. 
Our second test for a practically relevant deviation from deviation from the average $\mu_0^{t_0}$    rejects the null hypothesis in \eqref{hypotheses}, whenever
\begin{align}
    \label{test2}
       \hat T_{n,\Delta}\geq q^*_{1-\alpha} 
    \end{align}
and the following result shows that this decision rule defines a consistent asymptotic level $\alpha$, which can detect local alternatives converging to null hypothesis at a parametric rate.


\begin{theorem}
\label{Thm:BootCons}
     Under assumptions  (A1) and (A2) the test \eqref{test2} is  consistent and has  asymptotic level $\alpha$. More precisely,
    \begin{align*} 
        \p(\hat T_{n,\Delta}\geq q^*_{1-\alpha})\rightarrow \begin{cases}
            0 \quad d_\infty<\Delta \\
            \alpha \quad d_\infty=\Delta \\
            1 \quad d_\infty>\Delta 
        \end{cases}
    \end{align*}
\end{theorem}

\begin{remark}
   {\rm 
\label{DeltaChoice}
 An important question from a practical point of view is the choice of the threshold $\Delta > 0 $, which has to be carefully discussed for each specific application. Essentially, this boils down to the important question when a deviation from the reference value $\mu_0^{t_0}$ is practically significant, which 
 is related to
the specification of the effect size \citep[see][]{cohen1988}. While in many situations, such as in the climate data example mentioned before,  this specification is quite obvious, there are 
 other applications where this  choice might be less clear. However,  for such cases it is possible to determine a threshold from the data which can serve as  measure of evidence  for a  deviation of $\mu$ from  the long term average $\mu_0^{t_0}$ with a controlled type I error  $\alpha$.
   
    To be precise, note that the hypotheses $H_0(\Delta_1)$ and $H_0(\Delta_2)$ in \eqref{hypotheses} are nested for $\Delta_1 < \Delta_2$ and  that the test statistic \eqref{teststat} is monotone in $\Delta$. As  the quantile  $q^*_{1-\alpha}$ does not depend on $\Delta$,  rejecting $H_0(\Delta)$ for $\Delta=\Delta_1$ also implies rejecting $H_0(\Delta)$ for all $\Delta<\Delta_1$. The sequential rejection principle then yields that we may simultaneously test the hypotheses \eqref{hypotheses} for different choices of $\Delta\geq 0$ until we find the minimum value, say  $\hat \Delta_\alpha$,  for which $H_0(\Delta)$ is not rejected, that is 
    \begin{align}
    \label{deltahat}
        \hat \Delta_\alpha:=\min \big \{\Delta \ge 0 \,| \, \hat T_{n,\Delta}\leq q^*_{1-\alpha} \big  \}~.
    \end{align}
    Consequently, one may postpone the selection of $\Delta$ until one has seen the data. The same arguments of course also hold for the more conservative procedure defined by \eqref{test1}.      
}  
\end{remark}

\section{Estimating the time of the first relevant deviation}
\label{sec5}

   \def\theequation{3.\arabic{equation}}	
   \setcounter{equation}{0}

If a relevant deviation from a benchmark has been detected, it is of interest to determine  the first time where this deviation occurs,  that is 
\begin{align}
\label{det000}
    t^*=\min\big\{t \in [t_0,1]\big| ~|~\mu(t)-\mu_0^{t_0}| \geq \Delta \big\}~.
\end{align}
A natural  estimator for $t^*$ is  the first time $k$ where at least one estimated difference $\hat \mu_0^{\lfloor nt_0 \rfloor}-\hat \mu_j^k$  exceeds  approximately $\Delta$,  and therefore  we define
\begin{align}
\label{Def:LocEst}
    \hat t=\min\Big\{k \geq \lfloor nt_0 \rfloor+c_n  ~\Big| ~\exists j \in \{\lfloor nt_0 \rfloor,\ldots  ,k-c_n\} \text{ such that }~ \\
    |\hat \mu_0^{\lfloor nt_0 \rfloor}-\hat \mu_j^k|\geq \Delta-\frac{\hat \sigma\log(n)}{\sqrt{k-j}} \Big\}
\end{align}
where the constant $c_n$ satisfies \eqref{p1} (here we define the minimum over an empty set as $\infty$). In the following discussion we  investigate the theoretical performance of this estimator, distinguishing the case where $t^*$ is a point of continuity  of $|\mu(t)-\mu_0^{t_0}|$ and where it is not. 

For this purpose, we introduce the function 
$$
t \to d(t)=|\mu(t)-\mu_0^{t_0}|
$$
and consider  first  the smooth setting. Intuitively,  a change at a point of continuity will be harder to detect if the function  $\mu$ is very flat at $t^*$. To  quantify this property,  we assume that there exists constants $\kappa, c_\kappa > 0 $ such that
\begin{align}
\label{smooth}
    \lim_{t \uparrow t^*}\frac{|d(t^*)-d(t)|}{(t^*-t)^\kappa}\rightarrow c_\kappa~.
\end{align}
To obtain an explicit convergence rate we also need an assumption to ensure that the function $d$ does not behave too irregularly close to the point $t^*$.
\begin{enumerate}
        \item[(A3)] For some constant $\gamma \in (0,t^*)$ the function $t \rightarrow \text{sign}(d(t))d(t)$ is increasing on the set 
        $$
    U_\gamma(t^*)=\big\{t \in [t_0,t^*]~\big| ~t>t^*-\gamma \big \}.
    $$
\end{enumerate}

\begin{theorem}
\label{thm31}
 Let Assumptions   (A1) - (A3) be satisified.
 \begin{itemize}
     \item[(a)]
 If $t^* \in (t_0,1]$,  condition \eqref{smooth} holds   and 
$   c_n^{\kappa+1/2}\lesssim n^\kappa\sqrt{\log(n)}$, 
then 
    \begin{align}
        \hat t=t^*+O_\P\Big(\frac{\log(n)}{\sqrt{c_n}}\Big)^{1/\kappa}~.
    \end{align}
     \item[(b)]    If $t^*=\infty$ we have $\P(\hat t=\infty)=1-o(1)$.
     \end{itemize}
\end{theorem}

Next we discuss the case, where there is a jump at the point $t^*$ and  assume for some $\epsilon>0$ that 
\begin{align}
\label{hd7}     |\mu(t)-\mu_0^{t_0}|\begin{cases}
        <\Delta-\epsilon &\text{ if }  ~t<t^*\\
        \geq \Delta &\text{ if }  ~ t=t^*\\
        \geq \Delta+O(t-t^*)
        & \text{ if } ~ 
        \tilde t>t>t^*
    \end{cases}
\end{align}
where 
$\tilde t > t^* $ is the smallest point with a jump  of the function $\mu$ at $\tilde t$ (if there are no  jumps for $t >t^*$ we set $\tilde t=1$). 
\begin{theorem}
\label{thm32}
 Let Assumptions   (A1) - (A3) be satisfied and let $c_n$ satisfy $c_n^{3/2} \lesssim n\sqrt{\log(n)}$. Then
 \begin{itemize}
     \item[(a)]
 If  $t^* \in (t_0,1]$ is a jump discontinuity satisfying \eqref{hd7}, then 
    \begin{align}
        \hat t=t^*+O_\P\Big(\frac{c_n}{n}\Big).
    \end{align}
    \item[(b)]    If $t^*=\infty$,  we have $\P(\hat t=\infty)=1-o(1)$.
    \end{itemize} 
\end{theorem}
Comparing the Theorem \ref{thm31} and \ref{thm32}, it is readily apparent that detecting smooth changes profits from a large $c_n$ (i.e. the mean is only estimated over longer intervals) while abrupt changes are easier detected if $c_n$ is chosen small (i.e. the mean is also estimated  over shorter intervals).

\section{Finite sample properties }
\label{sec6}

   \def\theequation{4.\arabic{equation}}	
   \setcounter{equation}{0}

In this section we investigate the finite sample properties of the proposed methodology  by means of  a simulation study and  illustrate its application by analyzing a real data example. For the sake of comparison we consider the same scenarios as investigated in \citep*{buecher21}.

\begin{figure}[t]
    \begin{center}
 \includegraphics[width=0.5\textwidth]{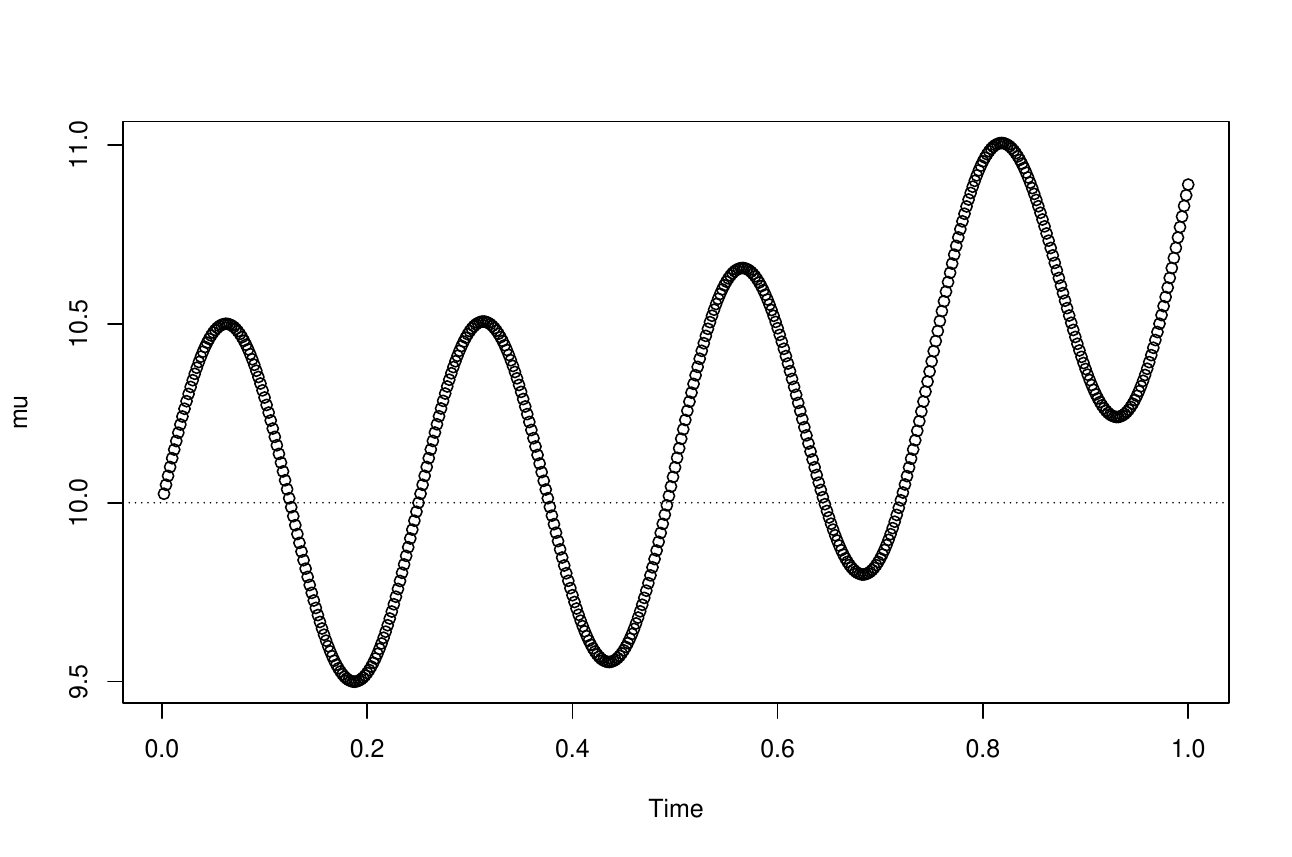}
  \caption{\it Plot of the regression function $\mu_a(x)$ in \eqref{meanfuncs} for $a=2$.  The dotted line is given by $\mu_0^{1/4} = 4\int_0^{1/4}\mu_2(s)ds$.}
  \label{figa}  
\end{center}
\end{figure}

\subsection{Synthetic Data}
We choose  $\Delta=1$ and  the mean function
\begin{align}
\label{meanfuncs}
    \mu_a(x)= 10+1/2\sin(8\pi x)+a\Big(x-\frac{1}{4}\Big)^2\1\Big\{x>\frac{1}{4}\Big\} , 
\end{align}
which is displayed in Figure \ref{figa} for $a=2$. We consider various choices of the parameter $a$ where we choose $t_0=1/4$  so that the hypotheses are given by
\begin{align}
\label{simhyp}
    H_0(1):d_\infty\leq 1 \quad vs \quad H_1(1): d_\infty>1
\end{align}
where
\begin{align}
    d_\infty=\sup_{t \in [1/4,1]}\Big| \mu_a(t)-4\int_0^{1/4}\mu_a(s)ds\Big|
\end{align}
Note that $d_\infty=1$ (boundary of the hypotheses)  for $a=\frac{128}{81}$ and that $d_\infty>1$ (alternative)   and $d_\infty<1$ (interior of the null hypothesis), whenever $a>\frac{128}{81}$ and  $ a<\frac{128}{81}$, respectively.  For the error processes $(\epsilon_i)_{i\in\Z}$ in model \eqref{model} we investigate the processes
\begin{align}
\label{indep}(\text{IID}) \quad &~ \epsilon_i= \tfrac{1}{2}\eta_i , \\
\label{MA}  (\text{MA)} \quad &~ \epsilon_i = \tfrac{1}{\sqrt{5}}\big( \eta_i + \tfrac{1}{2}\eta_{i-1}\big)\\
\label{AR} (\text{AR})  \quad &~ \epsilon_i = \tfrac{\sqrt{3}}{4} \big( \eta_i + \tfrac{1}{2}\epsilon_{i-1}\big),
\end{align}
where $(\eta_i)_{i\in\Z}$ is an i.i.d.\ sequence of standard normally distributed random variables. In particular, we have $\text{Var}(\epsilon_i)=\tfrac{1}{4}$ for all error processes under consideration. We will compare the novel testing procedures \eqref{test1} and \eqref{test2} proposed in this paper 
 with the most powerful test from \citep*{buecher21} which is given in equation (4.6) therein. Throughout this section  we generically choose $m=5$ (tuning parameter for the long run variance estimation)  and $c_n=20$ (lower bound for scales in the multi-scale statistic \eqref{teststat}).
The results are fairly stable under perturbation of these parameters as long as they are not chosen too small. 
The empirical rejection rates are calculated by  $1000$ simulation  runs. For the test \eqref{test1} we calculated the quantiles of the distribution of $ M$
by $1000$ samples from a Brownian motion sampled on a grid with width $0.001$. For the test \eqref{test2} we used $200$  samples to calculate the quantile $q_{1-\alpha^*}$ in \eqref{test1} for each of the $1000$ simulation runs.

The empirical rejection probabilities of all three tests are recorded  in Table \ref{tab1} and confirm the  asymptotic theory. Regarding the interpretation of the empirical findings we note that the null hypothesis  in \eqref{simhyp} is true, whenever the parameter $a$ satisfies $a\leq 128/81 \simeq 1.58$ and that we expect an increasing number of rejections for larger values of $a$, which yield increasing values for the difference $d_\infty-\Delta$.
Note that all three tests are conservative in the sense that the empirical size is smaller than $5\%$ at the boundary of the hypotheses \eqref{hypotheses} defined by $d_\infty-\Delta =0$  (in boldface).  However, it is worthwhile to mention that  the test \eqref{test2} provides a better approximation of  the nominal  level  than its competitors.  \\

{\footnotesize
\begin{table}[H]
\centering
			\begin{tabular}{l|c | ccc | ccc | ccc  }		\hline \hline
		$\mu_a$&test  &\multicolumn{3}{c|}{\eqref{test1} }&\multicolumn{3}{c|}{\eqref{test2} }&\multicolumn{3}{c|}{ \cite{buecher21}} \\
			$a$ &$d_\infty-\Delta$ & 200 & 500 & 1000 & 200 & 500 & 1000 & 200 & 500 & 1000  \\ 
		 \hline 
		\addlinespace[.2cm]
		\multicolumn{11}{l}{\quad\textit{Panel A: iid errors}} \\ 
		1.5 & -0.03& 0.0 &0.0 &0.0 &0.0 & 0.0 & 0.0 & 0.0 & 0.0 & 0.0  \\ 
		\bf 1.58 & \bf 0.00 & \bf 0.0  & \bf 0.1& \bf 0.2 &\bf 0.2 & \bf 1.5 & \bf 0.2 & \bf 0.0 & \bf 0.0 & \bf 0.2\\ 
		2.0 & 0.13 &  0.6& 2.5 & 3.2&5.5 & 14.1 & 15.4 & 0.0 & 3.3 & 23.1 \\ 
		2.5 & 0.29 & 10.4 & 42.6& 50.3&42.1 & 73.6 & 80.2 & 0.0 & 29.9 & 97.8 \\ 
		3.0 & 0.45 &  54.3& 93.8& 99.7&86.6 & 99.4 & 100 & 0.2 & 57.3 & 100  \\
	\addlinespace[.2cm]
		\multicolumn{11}{l}{\quad\textit{Panel B: MA errors}} \\ 
	1.5 & -0.03 &  0.1& 0.1& 0.2 &0.0 & 0.0 & 0.0 & 0.0 & 0.0 & 0.0 \\ 
	\bf 1.58 & \bf 0.0 & \bf 0.1  & \bf 0.6 &\bf 0.1&\bf 0.7 & \bf 2.7 & \bf 2.4 & \bf 0.0 & \bf 0.0 & \bf 0.3 \\ 
	2.0 & 0.13 & 1.2 & 5.1 & 4.3& 6.3 & 14.9 & 16.1 & 0.0 & 3.7 & 18.7 \\ 
	2.5 & 0.29 &  10.3& 31.0& 47.1&26.7 & 56.8 & 75.6 & 0.2 & 27.0 & 87.9\\ 
	3.0 & 0.45 &  44.9& 81.4& 96.6&69.7 & 94.5 & 99.7 & 0.5 & 52.8 & 99.7 \\ 
 \addlinespace[.2cm]
		\multicolumn{11}{l}{\quad\textit{Panel C: AR errors}} \\ 
	1.5 & -0.03 &  0.7& 1.1& 1.1&0.0 & 0.0 & 0.0 & 0.1 & 0.5 & 1.0\\ 
	\bf 1.58 & \bf0.00  & \bf 0.8 & \bf 1.1&\bf 2.0 &\bf 2.6 & \bf 5.9 & \bf 6.7 & \bf 0.1 & \bf 1.4 & \bf 1.5 \\ 
	2.0 & 0.13 & 3.1 & 9.4 & 10.7&7.8 & 22.6 & 27.0 & 0.0 & 7.8 & 23.1 \\ 
	2.5 & 0.29 &  15.6& 36.1& 51.1&29.6 & 60.6 & 75.1 & 0.4 & 27.3 & 77.7\\ 
	3.0 & 0.45 &  43.5& 81.2  & 95.9&67.3 & 91.7 & 99.5 & 1.1 & 53.9 & 98.4 \\ \hline \hline
	\end{tabular}  \smallskip
\caption{\it Empirical rejection rates of the tests \eqref{test1} and \eqref{test2} and the test  proposed in equation (4.6) of \citep*{buecher21} for the hypotheses  \eqref{simhyp}. Different values for the parameter~$a$ in the mean function \eqref{meanfuncs},  error processes, and sample sizes $n=200, 500, 1000$ are considered.  }  
\label{tab1}
\end{table}
}

A comparison of the power properties of the different procedures shows that the conservative multiscale test \eqref{test1} outperforms the test in \citep*{buecher21} for moderates samples sizes ($n=200,500$), but the last-named test yields larger rejection probabilities if the sample  size $n=1000$, in particular if the deviation is $d_\infty-\Delta =2$.
On the other hand 
the multi-scale test \eqref{test2} 
yields an even larger power for sample size $n=200, 500$,  and for  $n=1000$ it shows a similar performance as the procedure in \citep*{buecher21}.

Finally,  we present a brief simulation study to investigate the performance of the estimator \eqref{Def:LocEst}  for the first time of a relevant deviation as defined in \eqref{det000}, where we use  $m=5$ and $c_n=20+n^{1/2}$. We consider the mean function  \eqref{meanfuncs}  with   $a=2$, thus the time of a first relevant deviation is $t^*=0.791$. 
The error process is given by \eqref{AR} and the  sample size $n=500$. 
In Figure \ref{Fig:LocHists} we display histograms 
of the estimator \eqref{Def:LocEst} and the estimator proposed in equation (5.2)
of  \citep*{buecher21} based on $10.000$ simulation runs.
We observe  that the   estimator introduced in \citep*{buecher21} 
does not detect a relevant deviation 
in more than 90\% of all cases for many of the considered settings while our estimator detects such a deviation almost always. However, due to the noise of the error process there exist also  cases where $\hat t$  underestimates the true point $t^*$ 
and delivers an estimate for the local maximum  of the function $\mu$ at $t=0.57$ which is the closest local peak to the time $t=0.8$ (here the deviation is $|\mu (0.57) - \mu_0^{1/4}| \simeq 0.695 < 1$.  

We  therefore refrain from a direct comparison  of the tow estimators and display 
in Table \ref{Fig:LocEst} the empirical bias, standard deviation and the detection rate of the estimator \eqref{Def:LocEst} proposed in this paper.   
\begin{figure}[H]
    \begin{center}
  \includegraphics[width=0.45\textwidth]{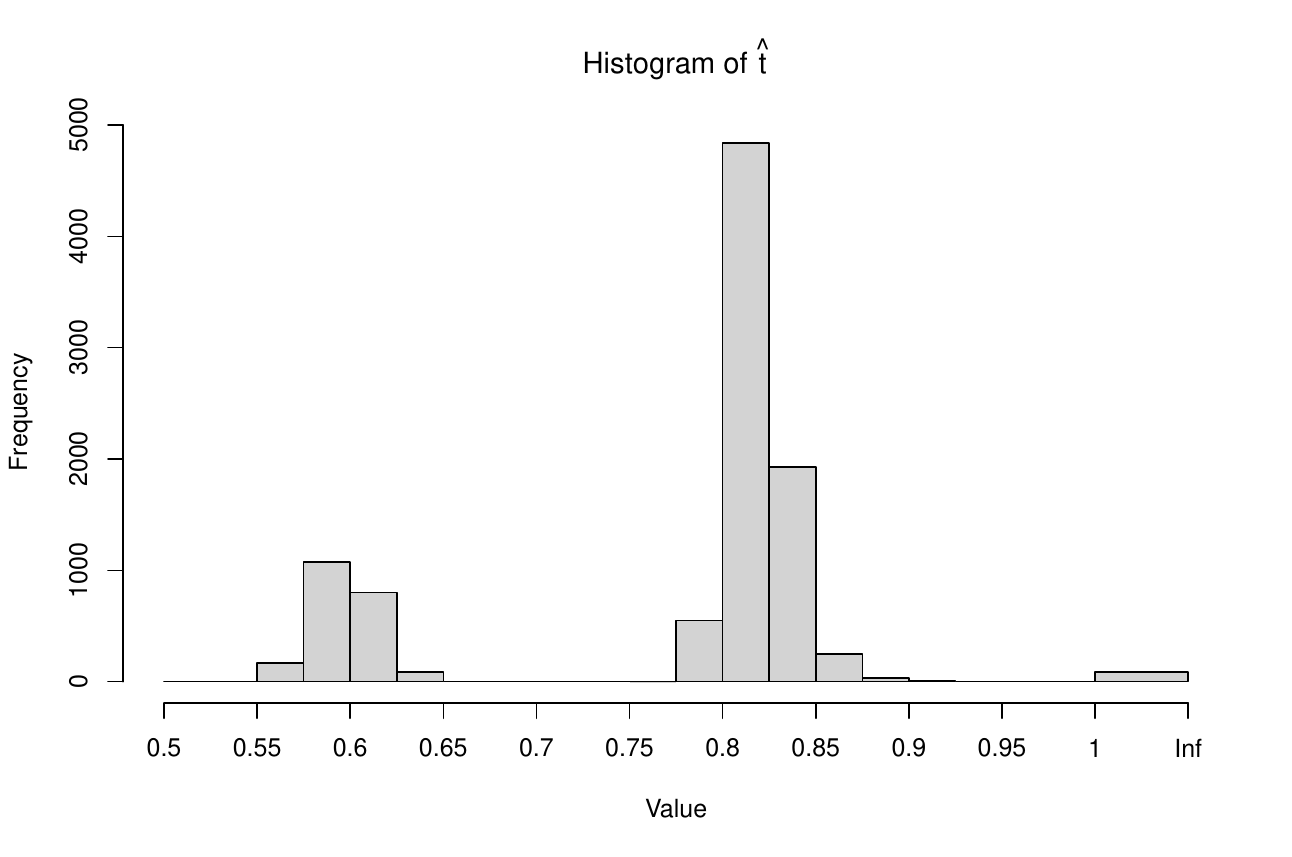}
  \includegraphics[width=0.45\textwidth]{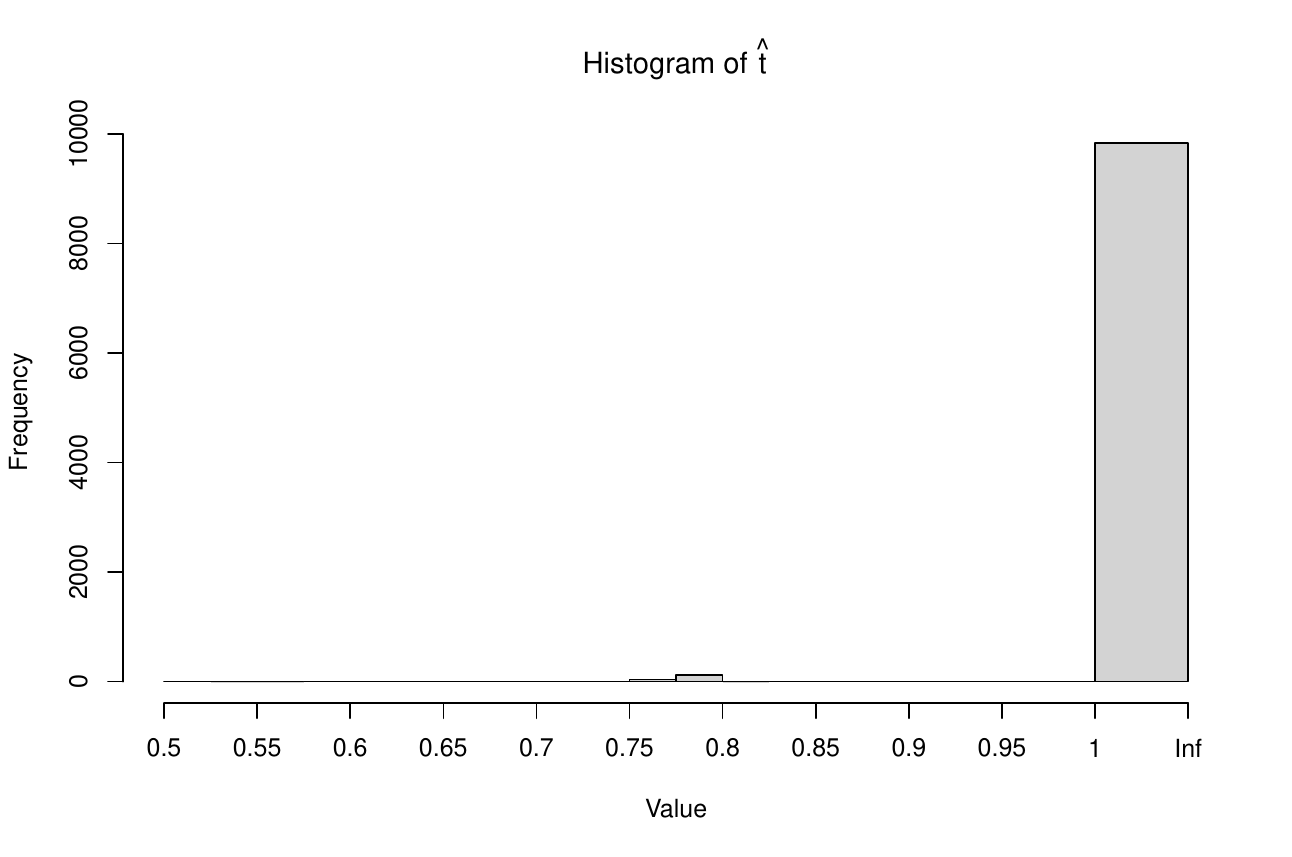}
  \caption{\it  Histograms for the estimator $\hat t$ defined  in \eqref{Def:LocEst} (left) and the estimator proposed in equation (5.2) of  \citep*{buecher21} (right).  
  \label{Fig:LocHists} }  
\end{center}
\end{figure}

Note that all choices of  the parameter $a$ in Table \ref{Fig:LocEst} correspond to  a violation of the null hypotheses. Therefore,  a ``good'' estimator of the time point of the first relevant deviation should  be finite in as many cases as possible. We observe that the estimator \eqref{Def:LocEst} proposed in this paper is finite almost always for all choices of $a$ even for the sample size $n=200$. 
The bias and standard deviation of the estimator first increase when the sample size increases to 500. This is due to the method sometimes detecting a change at the second highest peak of the the curve (see Figure \eqref{Fig:LocHists}) when the sample size becomes larger.
This happens more rarely when the sample size grows further, as is reflected by a lower bias for $n=1000$ compared to $n=500$. The estimator generally performs better for less dependent data but nonetheless performs well across all settings.

\begin{table}[H]
    \centering
    \begin{tabular}{lccc ccc}
        \toprule
        &  \multicolumn{3}{c}{\(\hat{E}[t^*]\)} & \multicolumn{3}{c}{\(\hat{E}[1(t^* = \infty)]\)} \\
        \cmidrule(lr){2-4} \cmidrule(lr){5-7}
        \(n \setminus a\)  & 2.0 & 2.5 & 3.0 & 2.0 & 2.5 & 3.0 \\
        \midrule
        \multicolumn{7}{l}{\textit{Panel A: iid errors}} \\        
        200 &  0.864 (0.029) &  0.839 (0.036) & 0.812 (0.052) & 0.031 & 0.002  & 0.000\\
        
        500 &  0.809 (0.044) & 0.780 (0.069) &  0.738 (0.090)& 0.001 & 0.000& 0.000  \\
        
        1000 &  0.803 (0.023)   & 0.783 (0.046)& 0.749 (0.076) &  0.000 & 0.000 &  0.000 \\
        \midrule
        \multicolumn{7}{l}{\textit{Panel B: MA errors}} \\
        
        200  &  0.854 (0.054) & 0.826 (0.063) & 0.794 (0.074) & 0.057 & 0.029 & 0.000  \\
        
        500  &  0.788 (0.078) & 0.754 (0.093) & 0.711 (0.101) & 0.001 & 0.000 & 0.000\\
       
        1000 & 0.789 (0.057) & 0.757 (0.080) & 0.715 (0.095) & 0.000 & 0.000 & 0.000 \\
        \midrule
        \multicolumn{7}{l}{\textit{Panel C: AR errors}} \\
        
        200  & 0.835 (0.084) & 0.810 (0.087) & 0.778 (0.090) & 0.094 & 0.011 & 0.001 \\
        
        500  & 0.764 (0.108) & 0.727 (0.116) & 0.693 (0.111) &0.009  & 0.001 & 0.000 \\
        
        1000 & 0.767 (0.087) & 0.731 (0.099) & 0.692 (0.104) & 0.002 & 0.000 & 0.000\\
        \midrule
        \( t^* \) & 0.79 & 0.78 & 0.77 & & & \\
        \bottomrule
    \end{tabular}
     \caption{\it Empirical bias, standard deviation and detection rate of the estimator  \( \hat t \) defined in \eqref{Def:LocEst}. Central part: empirical mean and standard deviation (in brackets) of \( \hat t \), conditional on \( t^* \neq \infty \). Right part: proportion of cases for which \( \hat t = \infty \). Last line: true change point \( t^* \). }
    \label{Fig:LocEst}
\end{table}

\subsection{Real Data Application}

We consider the mean of daily minimal temperatures (in degree Celsius) over the month of July for different weather stations in Australia.  The data set is available via the R package fChange \citep[see][]{Sonmez} on Github and the sample size varies between $100$ and $150$, depending on the weather station.
For each weather station we test the hypotheses \eqref{hypotheses} for  different thresholds $\Delta \in \{0.5,1,1.5\}$, where $t_0$ is chosen such  that the years $1,...,nt_0$ correspond to the time frame until the year $1950$. In Table \ref{tab:app} we  record the $p$-values of the multi-scale test \eqref{test2} and the test proposed in equation (4.6) of  \citep*{buecher21}. For the test \eqref{test2} these  $p$-values are calculated by $1000$  bootstrap repetitions, while the parameters $m$ and $c_n$ are  chosen as in the previous section, that is $c_n=20$ and $m=5$.

{\small
\begin{table}[H]
\centering
	\begin{tabular}{l | ccc | ccc  } \hline \hline
		$\Delta$ & \multicolumn{1}{c}{0.5} & \multicolumn{1}{c}{1.0} & \multicolumn{1}{c|}{1.5} & \multicolumn{1}{c}{0.5} & \multicolumn{1}{c}{1.0} & \multicolumn{1}{c}{1.5} \\
		\hline
        \hline
        test & \multicolumn{3}{c|}{\cite{buecher21}} & \multicolumn{3}{c}{ \eqref{test2}}\\
        \hline
        \hline
		Boulia/p-value & 29.0 & 73.1 & 98.0 & \textbf{0.0} & \textbf{0.8}& 37.7 \\
        Boulia/year &- & - & - & 1950& 1956 & -\\
        \hline
		Cape Otway/p-value & 11.4 & 98.0 & 100.0 & 7.4 & 78.8 &  99.9 \\
        Cape Otway/year & - & - & - & - & - & - \\
        \hline
		Gayndah/p-value & \textbf{0.2} & \textbf{1.4} & \textbf{3.2} & \textbf{0.0} & \textbf{0.0}& \textbf{0.3} \\
        Gayndah/year & 1952 & 1968 & 1974 & 1950 & 1950 & 1984\\
                \hline
		Gunnedah/p-value & \textbf{0.8} & \textbf{2.7} & \textbf{10.3} & \textbf{0.0} & \textbf{0.0} & \textbf{1.4}\\
        Gunnedah/year & 1952 & 1955 & - & 1962 & 1973 & 1984\\ 
        \hline
		Hobart/p-value & 94.9 & 100.0 & 100.0 & 34.1  & 95.5 & 100.0 \\
        Hobart/year & - & - & - & - & - & - \\
        \hline
		Melbourne/p-value & \textbf{0.0} & \textbf{1.1} & 27.3 & \textbf{0.0} & \textbf{0.0}& \textbf{3.7} \\
        Melbourne/year & 1968 & 1976 & - & 1950 & 1972 & 1990\\
        \hline
		Robe/p-value & \textbf{3.3} & 44.9 & 98.2 & 6.1 & 68.3 & 99.8 \\
        Robe/year & 1953 &- &- & - & - & - \\
              \hline
		Sydney/p-value & 41.1 & 98.1 & 100 & \textbf{0.05}& 0.14 &  89.9 \\	
        Sydney/year & - & - & - & 1950 & - & - \\
		\hline \hline
	\end{tabular} \smallskip
	\caption{\it $p$-values and estimates for $t^*$ of the test in \cite{buecher21} (left part) and the 
    bootstrap test \eqref{test2}  proposed in this paper (right part) for the hypotheses \eqref{hypotheses}   for various values of the threshold  $\Delta$. }
	 	\label{tab:app}
\end{table}
}

Except for the Robe weather station the $p$ values of the multi-scale test test \eqref{test2} are either similar  or substantially smaller  than the $p$ values obtained by  the test in \citep*{buecher21}. In particular the new test detects changes in Boulia, Melbourne and Sydney that the procedure from \citep*{buecher21} was not able to identify. We also observe that in general the new estimator $\hat t$  proposed in this paper generally dates deviations earlier than its counterpart from \citep*{buecher21}.
The only exception is  the station Robe, where the test from \citep*{buecher21} detects a difference of at least $0.5$ degrees Celsius  that our method does not detect at significance level $\alpha=0.05$. However, a more precise look at the minimum value  $\hat \Delta_\alpha$,  for which $H_0(\Delta)$ is not rejected at a controlled type I error $\alpha$ (see equation \eqref{deltahat} and equation (5.2) in \citep*{buecher21})
shows  that the difference between the two tests are small: the method proposed in  \citep*{buecher21} gives $\hat \Delta_{0.05} = 0.53 $, while the estimator \eqref{deltahat} yields
$\hat \Delta_{0.05} = 0.49$.

These values are taken from 
Table \ref{tab:delta}, which 
displays the values $\hat \Delta_{0.05}$ for both methods (here we have recalculated the results of the test of \citep*{buecher21}). The results further confirm the previous findings. Except for the weather station in Robe the test \eqref{test2} always detects larger differences than the test  proposed in \citep*{buecher21}. In particular we are able to detect changes in Boulia and Hobart where the method  \citep*{buecher21} does not detect any  relevant deviation. While at Hobart   the difference in the value $\hat \delta_{0.05}$ is small, it is larger than $1$ degree Celsius at Boulia.

{\small 
\begin{table}[H]
\centering
	\begin{tabular}{l | c |c |c | c |c |c | c | c  } \hline \hline
		$\hat \Delta_{0.05} $& Boulia & Cape Otway & Gayndah & Gunnedah & Hobart & Melbourne & Robe &Sydney \\
		\hline
        \hline
		BüDH & 0.00 & 0.38 & 1.69 & 1.22 & 0.00 & 1.17 & 0.53 & 0.13 \\
        BaD & 1.16 & 0.45 & 1.74 & 1.64 & 0.17 & 1.52 & 0.49 & 0.87\\
		\hline \hline
	\end{tabular} 
    \smallskip
	\caption{\it The minimum value  $\hat \Delta_{0.05}$,  for which $H_0(\Delta)$ is not rejected at a controlled type I error of $5\%$ (see equation \eqref{deltahat} and equation (5.2) in \cite{buecher21}). }
	 	\label{tab:delta}
\end{table}
}

\textbf{Acknowldgements: } This research was partially funded in the course of TRR 391 Spatio-temporal Statistics for the Transition of Energy and Transport (520388526) by the Deutsche Forschungsgemeinschaft (DFG, German Research Foundation).

\pagebreak

\section{Proofs}

   \def\theequation{5.\arabic{equation}}	
   \setcounter{equation}{0}

\label{proofs}

Throughout this section we use the  notation $k_0=\lfloor nt_0\rfloor$ and define for any constant $a>0$ the sets
\begin{align}
    \mathcal{E}_c(a)&=\Big\{ (j,k) \in \{k_0,...,n\}^2 \Big| k-j=c, \Big|\mu_0^{t_0}-\mu_{j/n}^{k/n}\Big|\geq d_\infty-a\log(n)/\sqrt{c} \Big\}
     \label{hd12a} \\
    \mathcal{E}_c^{co}(a)&=\Big\{ (j,k) \in \{k_0,...,n\}^2 \Big| k-j=c\Big\} \setminus \mathcal{E}_c (a)~.
    \label{hd12b}
\end{align}
We will generally suppress the dependence on $a$ in the notation except for the subsection about the bootstrap procedure. In the following we shall assume that $\sigma=1$, the general case follows by simple rescaling.

\subsection{Proof of Theorem \ref{Thm:Nulldist}} \label{sec51}

Define
\begin{align}
    \bar \mu_j^k&=\frac{1}{k-j}\sum_{i=j+1}^k\mu(i/n)    
\end{align}
and denote  by $\check B(s)=n^{-1/2}B(ns)$ a rescaled version of the Brownian motion $B$. We will start stating four auxiliary results, which will be used in the proof of Proof of Theorem \ref{Thm:Nulldist} and  of other results.  The proofs are given at  the end of this section.

\begin{Lemma}
\label{intapprox}
Grant assumption (A2). It then holds that
    \begin{align}
        \sup_{1 \leq j<k\leq n}(k-j)(\bar \mu_j^k-\mu^{k/n}_{j/n})=O(1)
    \end{align}      
    
\end{Lemma}

\begin{Lemma}
\label{Lemma:Upper}
If  $d_\infty>0$ and  assumption   (A1) and (A2) are satisfied,  we have with high probability that
    \begin{align}
         \hat  T_n \leq  \sup_{\substack{c_n \leq c \leq n-k_0 \\ c \in \N }}\sup_{(j,k) \in \mathcal{E}_c}\mathfrak{s}(j/n,k/n)\Big(\sqrt{c}\frac{B(k_0)}{k_0}-\frac{B(k)-B(j)}{\sqrt{c}}\Big)- \Gamma_n(c)\Big)
    \end{align}
    where $B$ is a Brownian motion with variance $\sigma^2$ and suprema over empty sets are defined as  as $-\infty$.
\end{Lemma}

\begin{Lemma}
\label{Lemma:epslimit}
If  $d_\infty>0$ and  Assumption   (A1) and (A2) aresatisfied,  we have
    \begin{align}
       &\sup_{\substack{c_n \leq c \leq n-k_0 \\ c \in \N }}\sup_{(j,k) \in \mathcal{E}_c}\mathfrak{s}(j/n,k/n)\Big(\sqrt{c}\frac{B(k_0)}{k_0}-\frac{B(k)-B(j)}{\sqrt{c}}\Big)- \Gamma_n(c)\\
       \leq&\lim_{\epsilon \downarrow 0}\sup_{(s,t) \in \mathcal{A}_{\epsilon,d_\infty}}\mathfrak{s}(s,t)\Big(\sqrt{t-s}\frac{\check B(t_0)}{t_0}-\frac{\check B(t)-\check B(s)}{\sqrt{t-s}}\Big)- \Gamma(t-s)+o_{a.s.}(1) , 
    \end{align}
    where 
    \begin{align}
        \mathcal{A}_{\epsilon,d_\infty}=\Big\{ (s,t) \in [t_0,1]^2 \Big| s<t, \Big|\mu_0^{t_0}-\mu_s^t\Big|\geq d_\infty-\epsilon\Big\}
    \end{align}
\end{Lemma}

\begin{Lemma}
\label{Lemma:Lower}
 If $d_\infty>0$ and Assumption   (A1) and (A2) are satisfied, we have 
\begin{align}
    &\lim_{\epsilon \downarrow 0}\sup_{(s,t) \in \mathcal{A}_{\epsilon,d_\infty}}\mathfrak{s}(s,t)\Big(\sqrt{t-s}\frac{\check B(t_0)}{t_0}-\frac{\check B(t)-\check B(s)}{\sqrt{t-s}}\Big)- \Gamma(t-s) \\
    \leq&\sup_{\substack{c_n \leq c \leq n-k_0 \\ c \in \N }}\sup_{(j,k) \in \mathcal{E}_c}\mathfrak{s}(j/n,k/n)\Big(\sqrt{c}\frac{B(k_0)}{k_0}-\frac{B(k)-B(j)}{\sqrt{c}}-\sqrt{c}\Big(\mu_0^{t_0}-\mu_{j/n}^{k/n}\Big)\Big)\\
         & \qquad \qquad \qquad \qquad \qquad - \Gamma_n(c)-\sqrt{c}d_\infty\Big)+o_\p(1)\\
         =&\hat  T_n+o_\p(1)
\end{align}    
\end{Lemma}
\textbf{Proof of Theorem \ref{Thm:Nulldist}}

The weak convergence of the statistic  $\hat  T_n$ in \eqref{hd11} follows  directly from Lemma \ref{Lemma:Upper} -  \ref{Lemma:Lower}. Regarding the continuity of the distribution of $T_{d_\infty}$ we note that $T_{d_\infty}$ is a limit of convex functions of a Gaussian process and therefore a convex function of a Gaussian process itself. The continutiy then follows by Theorem 4.4.1 from \citep*{bogachev2015}.

The statements (1) and (2) regarding the asymptotic properties of the test statistic  $\hat T_{n,\Delta}$  are a direct consequence of \eqref{hd11} 
if  $d_\infty \leq \Delta$. 
In the case  $d_\infty>\Delta$ we note that piecewise Lipschitz continuity of $\mu$ yields that there exists a sequence of $j_n,k_n$ with $k_n-j_n \simeq n$ such that for some $\rho>0$ we have $|\mu(i/n)-\mu_0^{t_0}|\geq \Delta+\rho$ for all $j_n \leq i \leq k_n$. Consequently, using Lemma \ref{intapprox}, we obtain
    \begin{align}
        \Big|\frac{1}{k_0}\sum_{i=1}^{k_0}\mu(i/n)-\frac{1}{c}\sum_{i=j_n+1}^{k_n}\mu(i/n)\Big| \geq \Delta+\rho-O(n^{-1})
    \end{align}
    which yields $\hat T_{n,\Delta}\rightarrow \infty$ by an application of the triangle  inequality.

\subsubsection{Proof of  Lemma \ref{intapprox} -  \ref{Lemma:Lower}}
\label{sec512}

\begin{proof}[\bf Proof of Lemma \ref{intapprox}]
Let us first assume that $\mu$ has no discontinuities, then
    \begin{align}
        \bar \mu_j^k-\mu^{k/n}_{j/n}&=\frac{1}{k-j}\sum_{i=j+1}^k\Big(\mu(i/n)-n\int_{(i-1)/n}^{i/n}\mu(t)dt\Big)\\
        &=\frac{1}{k-j}\sum_{i=j+1}^kn\int_{(i-1)/n}^{i/n}\mu(t)-\mu(i/n)dt\\
        &\lesssim \frac{1}{k-j}\sum_{i=j+1}^k n^{-1}=n^{-1}
    \end{align}
The general case follows by splitting up the integrals containing the discontinuities, leading to finitely many additional terms in the sum that can be bounded only by a constant instead of $n^{-1}$.
\end{proof}

\begin{proof}[\bf  Proof of Lemma \ref{Lemma:Upper}]
By Assumption   (A1) and condition  \eqref{p1} we have
\begin{align}
   \sup_{n\geq k-j\geq c_n} \sqrt{k-j}|\hat \mu_0^{k_0}-\hat \mu_j^k|&=\sup_{|k-j|\geq c_n} \frac{|B(k)-B(j)|}{\sqrt{k-j}}+\frac{O(n^{1/2-p})}{\sqrt{k-j}}\\
   &=\sup_{|k-j|\geq c_n} \frac{|B(k)-B(j)|}{\sqrt{k-j}}+o_\P(1)~.
\end{align}
Using this and Lemma \ref{intapprox} we therefore obtain that
    \begin{align}
         \hat  T_n&= \sup_{\substack{c_n \leq c \leq n-k_0 \\ c \in \N }}\sup_{\substack{|k-j|=c \\ k>j\geq k_0}}\Big(\Big|\sqrt{c}\frac{B(k_0)}{k_0}-\frac{B(k)-B(j)}{\sqrt{c}}-\sqrt{c}\Big(\frac{1}{k_0}\sum_{i=1}^{k_0}\mu(i/n)-\frac{1}{c}\sum_{i=j+1}^k\mu(i/n)\Big) \Big| \\
         & \qquad \qquad \qquad \qquad \qquad - \Gamma_n(c)-\sqrt{c}d_\infty\Big)+o_\p(1)\\
         &= \sup_{\substack{c_n \leq c \leq n-k_0 \\ c \in \N }}\sup_{\substack{|k-j|=c \\ k>j\geq k_0}}\Big(\Big|\sqrt{c}\frac{B(k_0)}{k_0}-\frac{B(k)-B(j)}{\sqrt{c}}-\sqrt{c}\Big(\mu_0^{t_0}-\mu_{j/n}^{k/n}\Big) \Big| \\
         & \qquad \qquad \qquad \qquad \qquad - \Gamma_n(c)-\sqrt{c}d_\infty\Big)+o_\p(1)
    \end{align}
    Now note that it follows from the discusssion in  Section 2.2 of  \citep*{Frick2014} that  the random variable $M$ defined in  \eqref{p4} is finite with probability $1$,  which implies
    \begin{align}
        \sup_{\substack{c_n \leq c \leq n-k_0 \\ c \in \N }}\sup_{\substack{|k-j|=c \\ k>j\geq k_0}} \Big|\sqrt{c}\frac{B(k_0)}{k_0}-\frac{B(k)-B(j)}{\sqrt{c}} \Big|-\Gamma_n(c)=O_\p(1).
    \end{align}
  This  yields
    \begin{align}
        &\sup_{\substack{c_n \leq c \leq n-k_0 \\ c \in \N }}\sup_{(j,k) \in \mathcal{E}_c^{co}}\Big(\Big|\sqrt{c}\frac{B(k_0)}{k_0}-\frac{B(k)-B(j)}{\sqrt{c}}-\sqrt{c}\Big(\mu_0^{t_0}-\mu_{j/n}^{k/n}\Big) \Big|\\
         & \qquad \qquad \qquad \qquad \qquad - \Gamma_n(c)-\sqrt{c}d_\infty\Big)\\        
        \lesssim& - \log(n)
    \end{align}
    with high probability by the definition of the set  $\mathcal{E}^{co}_c$ in \eqref{hd12b}. Therefore, 
    \begin{align}
        &\sup_{\substack{c_n \leq c \leq n-k_0 \\ c \in \N }}\sup_{\substack{|k-j|=c \\ k>j\geq k_0}}\Big(\Big|\sqrt{c}\frac{B(k_0)}{k_0}-\frac{B(k)-B(j)}{\sqrt{c}}-\sqrt{c}\Big(\mu_0^{t_0}-\mu_{j/n}^{k/n}\Big) \Big|\Big)\\
         & \qquad \qquad \qquad \qquad \qquad - \Gamma_n(c)-\sqrt{c}d_\infty\Big)\\
        =&\sup_{\substack{c_n \leq c \leq n-k_0 \\ c \in \N }}\sup_{(j,k) \in \mathcal{E}_c}\Big(\Big|\sqrt{c}\frac{B(k_0)}{k_0}-\frac{B(k)-B(j)}{\sqrt{c}}-\sqrt{c}\Big(\mu_0^{t_0}-\mu_{j/n}^{k/n}\Big) \Big|\Big)\\
         & \qquad \qquad \qquad \qquad \qquad - \Gamma_n(c)-\sqrt{c}d_\infty\Big)+o_\p(1)\\
         =&  \sup_{\substack{c_n \leq c \leq n-k_0 \\ c \in \N }}\sup_{(j,k) \in \mathcal{E}_c}\mathfrak{s}(j/n,k/n)\Big(\sqrt{c}\frac{B(k_0)}{k_0}-\frac{B(k)-B(j)}{\sqrt{c}}-\sqrt{c}\Big(\mu_0^{t_0}-\mu_{j/n}^{k/n}\Big)\Big)\\
         & \qquad \qquad \qquad \qquad \qquad - \Gamma_n(c)-\sqrt{c}d_\infty\Big)+o_\p(1)\\
         \leq & \sup_{\substack{c_n \leq c \leq n-k_0 \\ c \in \N }}\sup_{(j,k) \in \mathcal{E}_c}\mathfrak{s}(j/n,k/n)\Big(\sqrt{c}\frac{B(k_0)}{k_0}-\frac{B(k)-B(j)}{\sqrt{c}}\Big)- \Gamma_n(c)
    \end{align}
    with high probability, which  yields the desired statement.
\end{proof}

\begin{proof}[\bf  Proof of Lemma \ref{Lemma:epslimit}]
Existence of the limit with respect to  $\epsilon$ follows because the quantity is, as a function on the probability space, pointwise monotonically non-increasing in $\epsilon$ and bounded because the random variable  $M$  is finite almost surely. The asymptotic inequality follows because $\bigcup_{c_n \leq c \leq n-k_0}\mathcal{E}_c$ is, for any $\epsilon>0$, eventually a subset of  $\mathcal{A}_{\epsilon,d_\infty}$.
\end{proof}

\begin{proof}[\bf Proof of Lemma \ref{Lemma:Lower}]
The equality has been established in the proof of Lemma \ref{Lemma:Upper} already. For the upper bound we proceed in three steps:
\begin{enumerate}
    \item[(I)]  For any sequence with $b_n=o(n^{-1/2})$ we have
    \begin{align}
        &\sup_{(j,k) \in \underset{c_n \leq c \leq n-k_0}{\bigcup}\mathcal{E}_c}\mathfrak{s}(j/n,k/n)\Big(\sqrt{k-j}\frac{B(k_0)}{k_0}-\frac{B(k)-B(j)}{\sqrt{k-j}}\\
         & \qquad \qquad \qquad \qquad \qquad -\sqrt{k-j}\Big(\mu_0^{t_0}-\mu_{j/n}^{k/n}\Big)\Big) -\Gamma_n(k-j)-\sqrt{k-j}d_\infty\Big) \\
        \geq & \sup_{ \substack{(s,t) \in A_{b_n}\cap\{k_0/n,...,1\}^2\\ |s-t|\geq c_n/n}}\mathfrak{s}(s,t)\Big(\sqrt{t-s}\frac{\check B(t_0)}{t_0}-\frac{\check B(t)-\check B(s)}{\sqrt{t-s}}-\sqrt{(t-s)n}\Big(\mu_0^{t_0}-\mu_{s}^{t}\Big)\Big)\\
         & \qquad \qquad \qquad \qquad \qquad - \Gamma(t-s)-\sqrt{(t-s)n}d_\infty\Big) 
    \end{align}
    by definition of the involved sets and of $\check B$.
    \item[(II)] By the definition of $A_{b_n}$ we then obtain
    \begin{align}
         &\sup_{ \substack{(s,t) \in A_{b_n}\cap\{k_0/n,...,1\}^2\\ |s-t|\geq c_n/n}}\mathfrak{s}(s,t)\Big(\sqrt{t-s}\frac{\check B(t_0)}{t_0}-\frac{\check B(t)-\check B(s)}{\sqrt{t-s}}-\sqrt{(t-s)n}\Big(\mu_0^{t_0}-\mu_{s}^{t}\Big)\Big)\\
         & \qquad \qquad \qquad \qquad \qquad - \Gamma(t-s)-\sqrt{(t-s)n}d_\infty\Big)  \\
         =&\sup_{ \substack{(s,t) \in A_{b_n}\cap\{k_0/n,...,1\}^2\\ |s-t|\geq c_n/n}}\Big(\mathfrak{s}(s,t)\Big(\sqrt{c}\frac{\check B(t_0)}{t_0}-\frac{\check B(t)-\check B(s)}{\sqrt{t-s}}\Big) - \Gamma(t-s)\Big) +o(1) \\
    \end{align}
    \item[(III)] Using similar arguments as in the proof of Theorem 2.1 in \citep*{Duembgen2002} (use Theorem 7.1 and Lemma 7.2 in \citep*{Duembgen2008} with  $\beta(x)=\1\{x \in [0,1]\}$ instead of Proposition 7.1 from \citep*{Duembgen2002})  we obtain 
    \begin{align}
    \label{pe:1}
        &\sup_{ \substack{(s,t) \in A_{b_n}\cap\{k_0/n,...,1\}^2\\ |s-t|\geq c_n/n}}\Big(\mathfrak{s}(j/n,k/n)\Big(\sqrt{c}\frac{\check B(k_0)}{k_0}-\frac{\check B(k)-\check B(j)}{\sqrt{c}}\Big) - \Gamma_n(c)\Big) \\
        =&\sup_{(s,t) \in A_{b_n}}\Big(\mathfrak{s}(j/n,k/n)\Big(\sqrt{c}\frac{\check B(k_0)}{k_0}-\frac{\check B(k)-\check B(j)}{\sqrt{c}}\Big) - \Gamma_n(c)\Big) +o_\p(1)
    \end{align}
\end{enumerate}
To be precise we proceed as in the proof of Theorem 2.1 in \citep*{Duembgen2002} to obtain, for any $\delta_1>0$, a set $A_1$ with probability $1-\delta_1$ on which there exists $\delta_2>0$ such that
\begin{align}
    &\sup_{ \substack{(s,t) \in A_{b_n}\cap\{k_0/n,...,1\}^2\\ |s-t|\geq c_n/n}}\Big(\mathfrak{s}(j/n,k/n)\Big(\sqrt{c}\frac{\check B(k_0)}{k_0}-\frac{\check B(k)-\check B(j)}{\sqrt{c}}\Big) - \Gamma_n(c)\Big) \\
    =&\sup_{ \substack{(s,t) \in A_{b_n}\cap\{k_0/n,...,1\}^2\\ |s-t|\geq \delta_2}}\Big(\mathfrak{s}(j/n,k/n)\Big(\sqrt{c}\frac{\check B(k_0)}{k_0}-\frac{\check B(k)-\check B(j)}{\sqrt{c}}\Big) - \Gamma_n(c)\Big) 
\end{align}
holds. Exactly the same arguments also yield a set $A_2$ with probability $1-\delta_1$ on which
\begin{align}
    &\sup_{(s,t) \in A_{b_n}}\Big(\mathfrak{s}(j/n,k/n)\Big(\sqrt{c}\frac{\check B(k_0)}{k_0}-\frac{\check B(k)-\check B(j)}{\sqrt{c}}\Big) - \Gamma_n(c)\Big)\\
    =&\sup_{\substack{(s,t) \in A_{b_n}\\|s-t|\geq \delta_2}}\Big(\mathfrak{s}(j/n,k/n)\Big(\sqrt{c}\frac{\check B(k_0)}{k_0}-\frac{\check B(k)-\check B(j)}{\sqrt{c}}\Big) - \Gamma_n(c)\Big)
\end{align}
holds. Equation \eqref{pe:1} then follows by a standard argument involving the uniform continuity of the process $(s,t)\rightarrow \frac{B(t)-B(s)}{\sqrt{t-s}}$ on the set $\{(s,t)| |s-t|\geq \delta_2\}$.\\

The Lemma then by the fact that 
$$
\sup_{(s,t) \in \mathcal{A}_{\epsilon,d_\infty}}\mathfrak{s}(s,t)\Big(\sqrt{t-s}\frac{\check B(t_0)}{t_0}-\frac{\check B(t)-\check B(s)}{\sqrt{t-s}}\Big)- \Gamma(t-s)
$$
is a decreasing function of  $\epsilon$.

\subsection{Proof of Theorem \ref{Thm:LocAlt}}

The proof follows by a straightforward modification  of Lemmas \ref{Lemma:Upper} to \ref{Lemma:Lower}. The key difference is that the quantity 
\begin{align}
    \sqrt{k-j}\Big(\Big(\mu_0^{t_0}-\mu^{k/n}_{j/n}\Big)-\Delta\Big)
\end{align}
is not upper bounded by (or in some cases converges to) $0$ anymore. Instead on uses the expansion
\begin{align}
    \Big(\mu_0^{t_0}-\mu^{k/n}_{j/n}\Big)-\Delta&=\Delta+\frac{\beta_n}{k/n-j/n}\int_{j/n}^{k/n}h(x)dx-\Delta\\
    &=\frac{\beta_n}{k/n-j/n}\int_{j/n}^{k/n}h(x)dx
\end{align}
in the last inequality in the proof of Lemma \ref{Lemma:Upper} and in step (II) of the proof of Lemma \ref{Lemma:Lower}.
\end{proof}

 \subsection{Proof of Theorem \ref{Thm:BootCons}}
We define the  quantity
\begin{align}
    \hat{\mathcal{E}}_c(a)=\Big\{ (j,k) \in \{k_0,...,n\}^2 \Big| k-j=c, \Big|\hat \mu_0^{k_0}-\hat\mu_{j}^{k}\Big|\geq \Delta-a\log(n)/\sqrt{c} \Big\}    ~.
\end{align}
and prove at the end of this section 
 the following auxiliary result.
\begin{Lemma}
\label{Lemma:Extremalsets}
    Let $\Delta=d_\infty$. For any fixed $a>0$ we have for $n$ large enough that
    \begin{align}
        \mathcal{E}_c(a-\epsilon) \subset \hat{\mathcal{E}}_c(a) \subset \mathcal{E}_c(a+\epsilon)
    \end{align}
    for all $c\geq c_n$.
\end{Lemma}

\begin{proof}[\bf Proof of Theorem \ref{Thm:BootCons}]
Let us first consider the case $d_\infty=\Delta$. By the consistency of the long run variance estimate  $\hat \sigma^2$ (see Section \ref{sec5.4} for a proof) we have that 
\begin{align}
    \hat{\mathcal{E}}_c(2\sigma) \subset  \hat{\mathcal{E}}_c \subset  \hat{\mathcal{E}}_c(\sigma/2)
\end{align}
with high probability. Lemmas \ref{Lemma:Extremalsets} and \ref{Lemma:epslimit} then yield the desired statement if we can show that $\hat{\mathfrak{s}}(j/n,k/n)=\mathfrak{s}(j/n,k/n)$ holds with high probability uniformly over $k,j$ with $(j/n,k/n) \in \mathcal{A}_{\epsilon,d_\infty}$ for some $ \epsilon<\Delta$. This is an easy consequence of Lemma \ref{intapprox} and equation \eqref{p3}.

For the other cases we note that the statistic $\hat T_n^*$ is stochastically bounded because the 
random variable $M$ defined in  \eqref{p4} is almost surely finite. As $\hat T_{n,\Delta} \rightarrow -\infty $ $(\infty)$ if $d_\infty<\Delta$ $(>\Delta)$ the assertion follows.
\end{proof}

\begin{proof}[\bf Proof of Lemma \ref{Lemma:Extremalsets}]
    By Assumption   (A1) and a union bound  it follows that the inequality 
    \begin{align}
    \label{p3}
        |\bar \mu_j^{k}-\hat\mu_{j}^{k}|&=\Big|(k-j)^{-1}\sum_{i=j+1}^k\epsilon_i\Big|\\
        &\lesssim \Big|(k-j)^{-1}(B(k)-B(j))\Big|+n^{1/2-p}/(k-j)\\
        &\lesssim \frac{\sqrt{\log(n)}}{\sqrt{k-j}}+n^{1/2-p}/(k-j)
    \end{align}
    holds  uniformly with respect to  $1 \leq k-j \leq n$  with high probability.
    As a consequence   we have
    \begin{align}
    \label{p8}
        &\p\Big(|\bar \mu_j^{k}-\hat\mu_{j}^{k}|  < \epsilon\log(n)/\sqrt{k-j}, c_n \leq k-j \leq n \Big)\\
        \geq &\p\Big( \frac{\sqrt{\log(n)}}{\sqrt{k-j}}+n^{1/2-p}/(k-j) < \epsilon \log(n)/\sqrt{k-j}, c_n \leq k-j \leq n\Big)~.
    \end{align}   
    Now condition  \eqref{p1} implies uniformly with resepect to  $c_n \leq k-j \leq n$ that
    \begin{align}
       { n^{1/2-p} \over k-j} \leq \frac{n^{1/2-p}}{\sqrt{c_n}\sqrt{k-j}}\lesssim \frac{o(1)}{\sqrt{k-j}}
    \end{align}    
    which then implies that the probability in \eqref{p8} is of order $1-o(1)$. This yields the desired set inclusions as they are true whenever
    \begin{align}
        |\bar \mu_j^{k}-\hat\mu_{j}^{k}|  < \epsilon\log(n)/\sqrt{k-j} ~,~~ c_n \leq k-j \leq n~. 
    \end{align}
\end{proof}

\subsection{Consistency of the long run variance estimate}
\label{sec5.4}
\begin{Lemma}
\label{Lemma:Var}
    Under assumption   (A1) and (A2) we have
    \begin{align}
        \hat \sigma^2=\sigma^2+O_\p(n^{-1/3})~.
    \end{align}
\end{Lemma}
\begin{proof}
    By assumption  (A1) we have that
    \begin{align}
        \hat \sigma^2=&\frac{1}{\lfloor n/m\rfloor -1}\sum_{i=1}^{\lfloor n/m \rfloor -1}\frac{\Big(2B(im)-B((i-1)m)-B((i+1)m)+A_i\Big)^2}{2m}      +o_\p(n^{-1/3})  
    \end{align}
    where
    \begin{align}
        A_i:=\bar \mu_{(i-1)m}^{im}-\bar \mu_{im}^{(i+1)m}~.
    \end{align}
    By Lemma \ref{intapprox} we have
    \begin{align}
        |A_i|&=\frac{n}{m}\Big|\mu_{(i-1)m/n}^{im/n}-\mu_{im/n}^{(i+1)m/n}\Big|+O(m^{-1})\\
        &\lesssim O(1)
    \end{align}
    where the inequality follows by the Lipschitz continuity of $\mu_s^t$ in $s$ and $t$. Standard arguments then yield
    \begin{align}
        \hat \sigma^2=&\frac{1}{\lfloor n/m\rfloor -1}\sum_{i=1}^{\lfloor n/m \rfloor -1}\frac{\Big(2B(im)-B((i-1)m)-B((i+1)m)\Big)^2}{2m}      +O_\p(n^{-1/3})  
    \end{align}
    which in turn yields the desired statement by noting that $Z_{im}=B(im)-B((i-1)m)$ is a triangular array of independent $\mathcal{N}(0,\sigma^2)$ variables.
\end{proof}

\subsection{Proof of Theorem \ref{thm31} and \ref{thm32}}

\textbf{Proof of Theorem \ref{thm31}}
    We only consider the case $t^*<\infty$, the case $t^*=\infty$ follows by easier and analogous arguments. We give an upper and a lower bound which establish the desired result upon combining them.
    
\textbf{Upper bound:}
We first note that it follows by Assumption  (A1), condition  \eqref{p1} and the fact that 
the random variable $M$ is 
almost surely  finite 
that 
\begin{align}
\label{p5}
    \sup_{c_n \leq |k-j|\leq n}\Big||\hat \mu_0^{k_0}-\hat \mu_j^k|-| \bar \mu_0^{k_0} - \bar \mu_j^k|\Big|= O_\P\Big(\sqrt{\frac{\log(n)}{c_n}}\Big)~.
\end{align}
We also note that the identity 
\begin{align}
    \mu(t^*)-\mu_0^{t_0}=\Delta
\end{align}
implies 
\begin{align}
    \bar \mu_{j_n}^{k_n}-\mu_0^{t_0}\geq\frac{1}{k_n-j_n}\sum_{i=j_n+1}^{k_n}(\Delta-C(t-i/n)^\kappa)=\Delta-\frac{C}{k_n-j_n}\sum_{i=j_n+1}^{k_n}(t-i/n)^\kappa
\end{align}
 for  $k_n=\lfloor nt\rfloor$ and $j_n=k_n-c_n$, 
where $C$ is some constant that only depends on $\mu$ and $c_\kappa$. We thus obtain
\begin{align}
\label{p6}
    \bar \mu_{j_n}^{k_n}-\mu_0^{t_0}\gtrsim \Delta-(c_n/n)^\kappa~, 
\end{align}
and a  similar argument is valid when $\mu(t)-\mu_0^{t_0}=-\Delta$.  Combining \eqref{p5} and \eqref{p6} therefore yields  
\begin{align}
\label{p12}
   |\hat \mu_0^{k_0}-\hat \mu_{j_n}^{k_n}|\gtrsim \Delta-(c_n/n)^\kappa-O_\P\Big(\sqrt{\frac{\log(n)}{c_n}}\Big)~, 
\end{align}
which implies that $\hat t\leq t^*$ holds with high probability.

\textbf{Lower bound:}
We give the argument for $\mu(t^*)-\mu_0^{t_0}=\Delta$, the other case follows analogously. By \eqref{smooth} we know that
\begin{align}
    \mu(t^*-x)=\Delta-c_\kappa x^\kappa+o(x^\kappa)~.
\end{align}
Thereby, choosing $x=\Big(\frac{3\sigma\log(n)}{c_\kappa\sqrt{c_n}}\Big)^{1/\kappa}$, we have
\begin{align}
    \mu(t^*-x)-\mu_0^{t_0}= \Delta-3\sigma\log(n)/\sqrt{c_n}+o\Big(\sqrt{\log(n)/c_n}\Big)
\end{align}
Consequently, using the continuity  of $\mu(t)-\mu_0^{t_0}$ and Assumption (A3) we obtain
\begin{align}
    \max_{t \in [t_0,t^*-x]}|\mu(t^*-x)-\mu_0^{t_0}|\leq \Delta-2\sigma\log(n)/\sqrt{c_n}
\end{align}
when $n$ is sufficiently large. Using similar arguments as in the derivation of the upper bound we therefore obtain
\begin{align}
\label{p7}
    \sup_{ k/n\in [t_0,t^*-x]}\sup_{j<k, k-j\geq c_n}|\hat \mu_0^{k_0}-\hat \mu_j^k|\leq \Delta-2\sigma\log(n)/\sqrt{c_n}+O_\P(\sqrt{\log(n)/c_n})~.
\end{align}
Now suppose that there exists $j_n<k_n$, satisfying $k_n-j_n\geq c_n, k_n/n \in [t_0,t^*-x]$, such that
\begin{align}
  |\hat \mu_0^{k_0}-\hat \mu_{j_n}^{k_n}|\geq \Delta-\frac{\hat \sigma \log(n)}{\sqrt{k_n-j_n}}
\end{align}
holds. By \eqref{p7} this would imply
\begin{align}
    \Delta-2\sigma\log(n)/\sqrt{c_n}+O_\P(\sqrt{\log(n)/c_n}) \geq \Delta-\frac{\hat \sigma \log(n)}{\sqrt{c_n}}
\end{align}
which happens only with probability $o(1)$ by Lemma \ref{Lemma:Var}. Consequently we have that $\hat t\geq t^*-x$ with high probability, i.e. 
\begin{align}
    \hat t\geq t^*-O_\P\Big(\frac{\log(n)}{\sqrt{c_n}}\Big)^{1/\kappa}~.
\end{align}
as desired.

\subsubsection{Proof of Theorem \ref{thm32}}
Again we only consider the case $t^*<\infty$ and note that the case $t^*=\infty$ follows by easier and analogous arguments. We give an upper and a lower bound which establish the desired result upon combination.\\
\textbf{Lower bound:}\\
Consider any $t<t^*$ and assume WLOG that $\mu(t^*)-\mu_0^{t_0}=\Delta$. Then, by \eqref{hd7}, we have
\begin{align}
    \bar \mu^k_j-\mu_0^{t_0}<\Delta-\epsilon
\end{align}
Equation \eqref{p5} then yields that $\hat t\geq t^*$ with high probability. \\
\textbf{Upper bound:}\\
Assume WLOG (last inquality of \eqref{hd7}) that for some $\delta>0$ we have for any $t^*\leq t \leq t^*+\delta$ that
\begin{align}
    \mu(t)-\mu_0^{t_0}\geq \Delta+O(t-t^*)~.
\end{align}
Consequently, by the same arguments leading to \eqref{p12}, we have
\begin{align}
   |\hat \mu_0^{k_0}-\hat \mu_{k_0}^{k_0+c_n}|\gtrsim \Delta-c_n/n-O_\P\Big(\sqrt{\frac{\log(n)}{c_n}}\Big)~, 
\end{align}
which yields
$\hat t \leq t^*+c_n/n$.

\bibliographystyle{apalike}
\setlength{\bibsep}{2pt}
\bibliography{main}

\begin{thebibliography}{}

\bibitem[Aston and Kirch, 2012]{Aston2012}
Aston, J. A.~D. and Kirch, C. (2012).
\newblock {Evaluating stationarity via change-point alternatives with applications to fMRI data}.
\newblock {\em The Annals of Applied Statistics}, 6(4):1906 -- 1948.

\bibitem[Aue and Horváth, 2013]{Aue2013}
Aue, A. and Horváth, L. (2013).
\newblock Structural breaks in time series.
\newblock {\em Journal of Time Series Analysis}, 34(1):1--16.

\bibitem[Aue and Kirch, 2023]{Aue2023}
Aue, A. and Kirch, C. (2023).
\newblock The state of cumulative sum sequential changepoint testing 70 years after page.
\newblock {\em Biometrika}, 111(2):367--391.

\bibitem[Baranowski et~al., 2019]{baranowskietal2019}
Baranowski, R., Chen, Y., and Fryzlewicz, P. (2019).
\newblock {Narrowest-Over-Threshold Detection of Multiple Change Points and Change-Point-Like Features}.
\newblock {\em Journal of the Royal Statistical Society Series B: Statistical Methodology}, 81(3):649--672.

\bibitem[Berkes et~al., 2011]{Berkes2011}
Berkes, I., H{\"o}rmann, S., and Schauer, J. (2011).
\newblock {Split invariance principles for stationary processes}.
\newblock {\em The Annals of Probability}, 39(6):2441 -- 2473.

\bibitem[Bogachev, 2015]{bogachev2015}
Bogachev, V. (2015).
\newblock {\em Gaussian Measures}.
\newblock Mathematical Surveys and Monographs. American Mathematical Society.

\bibitem[B{\"u}cher et~al., 2021]{buecher21}
B{\"u}cher, A., Dette, H., and Heinrichs, F. (2021).
\newblock Are deviations in a gradually varying mean relevant? a testing approach based on sup-norm estimators.
\newblock {\em The Annals of Statistics}, 49(6):3583 -- 3617.

\bibitem[Chen et~al., 2022]{CWS22}
Chen, Y., Wang, T., and Samworth, R.~J. (2022).
\newblock High-dimensional, multiscale online changepoint detection.
\newblock {\em J. R. Stat. Soc. Ser. B. Stat. Methodol.}, 84(1):234--266.

\bibitem[Cho and Kirch, 2024]{cho:kirch:2024}
Cho, H. and Kirch, C. (2024).
\newblock Data segmentation algorithms: Univariate mean change and beyond.
\newblock {\em Econometrics and Statistics}, 30:76--95.

\bibitem[Cohen, 1988]{cohen1988}
Cohen, J. (1988).
\newblock {\em Statistical Power Analysis for the Behavioral Sciences (2nd ed.)}.
\newblock Hillsdale, NJ: Lawrence Erlbaum Associates, Publishers.

\bibitem[Dehling, 1983]{Dehling1983}
Dehling, H. (1983).
\newblock Limit theorems for sums of weakly dependent banach space valued random variables.
\newblock {\em Zeitschrift f{\"u}r Wahrscheinlichkeitstheorie und Verwandte Gebiete}, 63:393--432.

\bibitem[Dette et~al., 2020]{Eckle2020}
Dette, H., Eckle, T., and Vetter, M. (2020).
\newblock {Multiscale change point detection for dependent data}.
\newblock {\em Scandinavian Journal of Statistics}, 47(4):1243--1274.

\bibitem[D{\"u}mbgen and Spokoiny, 2001]{Duembgen2001}
D{\"u}mbgen, L. and Spokoiny, V.~G. (2001).
\newblock {Multiscale Testing of Qualitative Hypotheses}.
\newblock {\em The Annals of Statistics}, 29(1):124 -- 152.

\bibitem[Dümbgen, 2002]{Duembgen2002}
Dümbgen, L. (2002).
\newblock Application of local rank tests to nonparametric regression.
\newblock {\em Journal of Nonparametric Statistics}, 14(5):511--537.

\bibitem[Dümbgen and Walther, 2008]{Duembgen2008}
Dümbgen, L. and Walther, G. (2008).
\newblock Multiscale inference about a density.
\newblock {\em The Annals of Statistics}, 36(4):1758--1785.

\bibitem[Frick et~al., 2014]{Frick2014}
Frick, K., Munk, A., and Sieling, H. (2014).
\newblock Multiscale change point inference.
\newblock {\em Journal of the Royal Statistical Society. Series B (Statistical Methodology)}, 76(3):495--580.

\bibitem[Fryzlewicz and Rao, 2013]{Fryzlewicz2013}
Fryzlewicz, P. and Rao, S. (2013).
\newblock Multiple-change-point detection for auto-regressive conditional heteroscedastic processes.
\newblock {\em Journal of the Royal Statistical Society: Series B (Statistical Methodology)}, 76.

\bibitem[Gao et~al., 2008]{Gao2008}
Gao, J., Gijbels, I., and {Van Bellegem}, S. (2008).
\newblock Nonparametric simultaneous testing for structural breaks.
\newblock {\em Journal of Econometrics}, 143(1):123--142.
\newblock Specification testing.

\bibitem[Gao et~al., 2018]{Gao2018}
Gao, Z., Shang, Z., Du, P., and Robertson, J. (2018).
\newblock Variance change point detection under a smoothly-changing mean trend with application to liver procurement.
\newblock {\em Journal of the American Statistical Association}, 114:1--29.

\bibitem[Gijbels et~al., 2007]{Gijbels2007}
Gijbels, I., Lambert, A., and Qiu, P. (2007).
\newblock Jump-preserving regression and smoothing using local linear fitting: A compromise.
\newblock {\em Annals of the Institute of Statistical Mathematics}, 59:235--272.

\bibitem[Horv\'ath and Rice, 2024]{Horvath24}
Horv\'ath, L. and Rice, G. (2024).
\newblock {\em Change point analysis for time series}.
\newblock Springer Series in Statistics. Springer Nature,.

\bibitem[Hotz et~al., 2013]{Hotz2013}
Hotz, T., Schütte, O.~M., Sieling, H., Polupanow, T., Diederichsen, U., Steinem, C., and Munk, A. (2013).
\newblock Idealizing ion channel recordings by a jump segmentation multiresolution filter.
\newblock {\em IEEE Transactions on NanoBioscience}, 12(4):376--386.

\bibitem[Karl et~al., 1995]{Karl1995}
Karl, T., Kniught, R., and Plummer, N. (1995).
\newblock Trends in high frequency climate variability in the twentieth century.
\newblock {\em Nature}, 377:217--220.

\bibitem[Li et~al., 2023]{LWR23}
Li, W., Wang, D., and Rinaldo, A. (2023).
\newblock Divide and conquer dynamic programming: An almost linear time change point detection methodology in high dimensions.
\newblock In {\em International Conference on Machine Learning}, pages 20065--20148. PMLR.

\bibitem[Madrid~Padilla et~al., 2022]{Padilla_Yu_Wang_Rinaldo_multivariate_nonparametric}
Madrid~Padilla, O.~H., Yu, Y., Wang, D., and Rinaldo, A. (2022).
\newblock Optimal nonparametric multivariate change point detection and localization.
\newblock {\em IEEE Trans. Inform. Theory}, 68(3):1922--1944.

\bibitem[Müller, 1992]{Muller1992}
Müller, H.-G. (1992).
\newblock {Change-Points in Nonparametric Regression Analysis}.
\newblock {\em The Annals of Statistics}, 20(2):737 -- 761.

\bibitem[Page, 1955]{Page1955}
Page, E.~S. (1955).
\newblock A test for a change in a parameter occurring at an unknown point.
\newblock {\em Biometrika}, 42(3-4):523--527.

\bibitem[Schmidt-Hieber et~al., 2013]{Schmidthieber2013}
Schmidt-Hieber, J., Munk, A., and D{\"u}mbgen, L. (2013).
\newblock {Multiscale methods for shape constraints in deconvolution: Confidence statements for qualitative features}.
\newblock {\em The Annals of Statistics}, 41(3):1299 -- 1328.

\bibitem[Sonmez et~al., 2025]{Sonmez}
Sonmez, O., Aue, A., and Rice, G. (2025).
\newblock {\em fChange: Change Point Analysis in Functional Data}.
\newblock R package version 0.2.0.

\bibitem[Truong et~al., 2020]{TRUONG2020}
Truong, C., Oudre, L., and Vayatis, N. (2020).
\newblock Selective review of offline change point detection methods.
\newblock {\em Signal Processing}, 167:107299.

\bibitem[Vogt and Dette, 2015]{Vogt2015}
Vogt, M. and Dette, H. (2015).
\newblock Detecting gradual changes in locally stationary processes.
\newblock {\em The Annals of Statistics}, 43(2):713--740.

\bibitem[Woodall and Montgomery, 1999]{woodmont1999}
Woodall, W.~H. and Montgomery, D.~C. (1999).
\newblock Research issues and ideas in statistical process control.
\newblock {\em Journal of Quality Technology}, 31(4):376--386.

\bibitem[Wu, 2005]{Wu2005}
Wu, W.~B. (2005).
\newblock Nonlinear system theory: Another look at dependence.
\newblock {\em Proceedings of the National Academy of Sciences}, 102(40):14150--14154.

\bibitem[Wu and Zhao, 2007]{Wu2007}
Wu, W.~B. and Zhao, Z. (2007).
\newblock Inference of trends in time series.
\newblock {\em Journal of the Royal Statistical Society: Series B (Statistical Methodology)}, 69(3):391--410.

\end{thebibliography}

\end{document}